\documentclass[12pt]{article}

\usepackage{amstext,amsmath,amssymb,amsfonts,bbm}
\usepackage[latin1]{inputenc}
\usepackage{epsfig}
\usepackage{dsfont}
\usepackage{hyperref}
\usepackage{amsthm}
\usepackage{subfigure}
\usepackage{color}
\usepackage{multirow}
\usepackage{psfrag}
\usepackage{graphicx}
\usepackage{extarrows}

\theoremstyle{plain}

\newtheorem{lemma}{Lemma}

\newtheorem{theorem}{Theorem}
\newtheorem{corollary}{ Corollary }
\newtheorem{definition}{Definition}
\newtheorem{proposition}{Proposition}
\theoremstyle{definition}

\DeclareMathOperator{\Tr}{Tr}

\begin{document}

\newcommand{\A}{A}
\newcommand{\M}{M}

\title{Analyticity results for the cumulants in a random matrix model}

\author{R. Gurau\footnote{rgurau@cpht.polytechnique.fr; CPHT - UMR 7644, CNRS,
 Ecole Polytechnique, 91128 Palaiseau cedex, France and Perimeter Institute for Theoretical 
Physics, 31 Caroline St. N, ON, N2L 2Y5, Waterloo, Canada.}  \,and
 T. Krajewski\footnote{thomas.krajewski@cpt.univ-mrs.fr; CPT - UMR 7332, CNRS, 
Aix-Marseille Universit\'e and Universit\'e de Toulon, Campus de Luminy, 13228 Marseille Cedex 9, France}}

\maketitle

\abstract{The generating function of the cumulants in random matrix models, as well as the cumulants themselves, can be expanded as asymptotic 
(divergent) series indexed by maps. While at fixed genus the sums over maps converge, the sums over genera do not. 
In this paper we obtain alternative  expansions both for the generating function and for the cumulants that cure this problem. 
We provide explicit and convergent expansions for the cumulants, for the remainders of their
perturbative expansion (in the size of the maps) and for the remainders of their topological expansion (in the genus of the maps). 
We show that any cumulant is an analytic function inside a cardioid domain in the complex plane and we prove that 
any cumulant is Borel summable at the origin.}

\tableofcontents

\section{Introduction}

Random matrix theory \cite{Mehta,Akemann} studies probability laws for matrices. They have been 
introduced more than half a century ago to model the energy spectra of large nuclei 
and have later proven to be ubiquitous in physics and mathematics. Applications to mathematics range from combinatorics of maps 
to free probability while in physics, beyond energy spectra of heavy nuclei, random matrices can be used to describe disordered 
systems and discretized models of random surfaces.  

The application of random matrices to random surfaces and 2d quantum 
gravity \cite{matrix} relies on the combinatorics of maps. The matrix integrals arising in random matrix
theory depend on (at least) two parameters: a coupling constant $\lambda$ and the size of the matrix, $N$.
A \emph{formal} expansion in the parameter $\lambda$ of such matrix integrals yields generating 
functions for maps of arbitrary genus. The coupling constant $\lambda$
measures the size of the map (the number of its edges), while the parameter $1/N$
turns out to measure the genus of the map. 
While this formal expansion is extremely successful 
in enumerating both maps of fixed size and arbitrary genus and maps of fixed genus and arbitrary size,
it does \emph{not} provide an estimation of the matrix integral because it \emph{does not converge}.

This phenomenon is well understood. From a combinatorial standpoint, the divergence of this formal series is due to the proliferation of maps:
while the maps of fixed genus are an exponentially bounded family, maps of arbitrary genus are not.
At the analytical level, this reflects the fact that $\lambda=0$ lies on the boundary of 
the analyticity domain of the generating function. 

One can analyze in some depth these formal power series. Restricting to  a fixed order in $1/N$
one obtains convergent series enumerating maps of fixed genus. This yields the celebrated $1/N$ expansion for random matrices \cite{'tHooft:1973jz}
(see also \cite{Alice,Albeverio,rigorous} for rigorous mathematical results on this expansion). 
The series at fixed genus  exhibit a critical behavior at 
some critical value  $\lambda_c$ of the coupling constant and a \emph{formal} sum over random surfaces of arbitrary genus can be obtained 
by taking the so called double scaling limit $\lambda \to \lambda_c, N\to \infty$ while keeping $(\lambda_c-\lambda)N^{5/4}$ fixed.
The precise status of the $1/N$ series is somewhat involved. Usually the $1/N$ series is taken as an \emph{asymptotic series}: while 
each order in $1/N$ is well understood, the rest term is usually difficult to control. 

Analytical control over the rest term has been 
achieved in the region of strictly convex potential (see for instance \cite{Alice2} and references therein). This region corresponds to a stable perturbation $\Re\lambda \ge 0$ or to 
an unstable but small perturbation (such that the perturbation potential is always dominated by the quadratic part and 
absolute converge of the matrix integral is ensured). However, one would like to exert some same kind of analytic control over the rest term 
of the $1/N$ series also \emph{outside} the region of strictly convex potential.
This is due to the following two facts:
\begin{itemize}
 \item when considering the 
interpretation of a matrix integral as a generating function of maps, $ \Re\lambda \ge 0$ corresponds to 
an \emph{alternating} sum over maps. A genuine sum over maps is obtained only for $ \Re\lambda < 0$. Moreover,
the critical point $\lambda_c$ (at which the fixed genus series become critical) lies far on the negative real 
axis (for bipartite quadrangulations for instance $\lambda_c=-\frac{1}{12}$). In order to study the behavior of the 
matrix integral in the critical regime one needs to control the rest term close to this critical point.
 \item in the region $ \Re\lambda < 0$, instanton effects $e^{-\frac{1}{|\lambda|}}$ are expected to play a very important role.
 As they correspond to non trivial solutions of the classical equations of motion, they correspond precisely to the region where the perturbation
  potential equals the quadratic part, hence outside the strictly convex potential region.
\end{itemize}

In this paper we focus on the analyticity of the cumulants in a specific random matrix model.
Building on results in random tensor theory \cite{Razvannonperturbative,tensornonperturbative}
and on the Loop Vertex Expansion  (LVE) introduced in \cite{LVE} we establish, for any cumulant,
an explicit expansion  which is convergent for $\lambda$ in the cardioid domain:
\[
  \lambda\in \mathbb{C} \;, \qquad  4|\lambda|< \cos^{2}\Big({\frac{\arg\lambda}{2}}\Big) \; .
\]
We further provide explicit convergent expressions for the remainder in the expansion in $\lambda$
and as a by product we prove that any cumulant is Borel summable in $\lambda$ uniformly in $N$.

More importantly, we provide explicit expressions for the remainder in the $1/N$ expansion of any cumulant
which is absolutely convergent in the cardioid domain:
\[
  \lambda\in \mathbb{C} \;, \qquad  12|\lambda|< \cos^{2}\Big({\frac{\arg\lambda}{2}}\Big) \; .
\]

We emphasize that this domain goes well outside the strictly convex potential region. Our paper is 
thus a first step towards the rigorous study of the instanton effects and of the critical regime in matrix models. However,
work still remains to be done: in order to access these effects, one needs to find analytic continuations of our
explicit formulae which hold all the way up to the negative real axis (and up to $-\frac{1}{12}$). 

This paper is divided into four parts. In section \ref{sec:statement} we introduce some notation and 
state our main results. 
In section \ref{sec:intfield} we introduce the intermediate field representation which we subsequently 
use for the proofs of our results which are performed in sections \ref{sec:proof1} and \ref{sec:proof2}.
Some technical details are collected in the appendices.

\section{Statement of the main results}\label{sec:statement}

\label{mainresultsec}
\paragraph{Matrix integral and normalization.}
In this paper, we consider the Gau\ss ian matrix model with a quartic perturbation. The generating function of its cumulants 
is defined by the integral over complex $N\times N$ matrices $M$:
\begin{multline}
{\cal Z}[J,J^{\dagger};\lambda,N]=\\
\frac{\int dM\exp\Big\{
-\Tr(MM^{\dagger})-\frac{\lambda}{2N}\Tr(MM^{\dagger}MM^{\dagger})
+\sqrt{N}\Tr(JM^{\dagger})+\sqrt{N}\Tr(MJ^{\dagger})\Big\}}{\int dM\exp\Big\{ -
\Tr(MM^{\dagger})-\frac{\lambda}{2N}\Tr(MM^{\dagger}MM^{\dagger})
\Big\}} \; .
\label{Mintegral}
\end{multline}

The source $J$ is itself a $N\times N$ complex matrix and $J^{\dagger}$  is its adjoint. The cumulants
of the quartic model are obtained by taking derivatives of $\log {\cal Z}$ with 
respect to $J$ and $J^{\dagger}$.

A Taylor expansion in $\lambda$, $J$ and $J^{\dagger}$, followed by the evaluation of the Gau\ss ian integral, expresses
${\cal Z}$ as a sum over ribbon Feynman graphs (or combinatorial maps). The normalization in $N$ has been chosen in such a way that the 
amplitude of a ribbon graph is $N^{\chi(G)}$, with $\chi(G)$ 
the Euler characteristic of the graph (i.e. the Euler characteristic of a surface of minimal number of handles in which $G$ can be embedded). 
Since we are only interested in the cumulants, we divide the integral by its value at $J=J^{\dagger}=0$.

The measure $dM$ is the standard Lebesgue measure on matrices suitably normalized
in such a way that ${\cal Z}[J,J^{\dagger};\lambda,N]=1$ for $\lambda=0$ and $J=J^{\dagger}=0$,
\begin{equation}
dM=\pi^{N}\prod_{1\leq i,j\leq N}d\text{Re}(M_{ij})d\text{Im}(M_{ij}) \; .
\end{equation}

 The analyticity of ${\cal Z}[J,J^{\dagger};N,\lambda]$ is fairly easy to establish using
 conventional techniques. However, in order to study the analyticity of $\log {\cal Z}[J,J^{\dagger};N,\lambda]$,
 these techniques have to be supplemented by a detailed study of the zeros of  ${\cal Z}[J,J^{\dagger}]$ 
 in the complex domain, which is a harder problem. 

\paragraph{Loop Vertex Expansion (LVE) graphs and their amplitudes}

The LVE is based on combinatorial maps with cilia. A \emph{cilium} is a half edge hooked to a vertex. 
A combinatorial map is a graph with a distinguished cyclic ordering of the half edges incident at each vertex.
Combinatorial maps are conveniently represented as \emph{ribbon graphs} whose vertices are disks and whose edges 
are ribbons (allowing one to encode graphically the ordering of the half edges incident at a vertex).

\begin{definition}[LVE graphs and corners] 
 A \emph{LVE graph} $(G,T)$ is a connected ribbon graph $G$ with labels on its vertices having furthermore:
 \begin{itemize}
  \item a distinguished spanning tree $T\subset G$.
  \item a labeling of the edges of $G$ not in $T$ (loop edges in physics parlance).
  \item at most one cilium per vertex. 
 \end{itemize}

A \emph{LVE tree} is a LVE graph without cycles.
 
A \emph{corner} of a LVE graph $(G,T)$ is a pair of consecutive half edges attached to the same vertex.  
\end{definition}

We denote $K(G)$, $V(G)$, $E(G)$ and $F(G)$ the sets of cilia, vertices, edges  and respectively faces of $G$. 
The edges of $G$ not in $T$ are called \emph{loop edges} and we denote $L(G,T)=E(G)-E(T)$ the set of 
loop edges. The faces of $G$ are partitioned between the faces which do not contain any cilium (which we sometimes call 
internal faces) and the ones which contain at least a cilium which we call \emph{broken faces}.  We denote $B(G)$ the set of broken faces of $G$.
Each broken face corresponds to a puncture in the Riemann surface in which $G$ is embedded, and 
the Euler characteristic of the graph $G$ is:
\begin{equation}
\chi(G)=|V(G)|-|E(G)|+|F(G)|-|B(G)|=2-2g(G)-|B(G)|
\label{Euler:eq}
\end{equation}
 where $|X|$ denotes the cardinality of $X$ and $g(G)$ is the genus of the graph $G$. 

Let us consider a LVE graph $(G,T)$ with vertices labeled $1,\dots |V(G)|$. We associate to every edge 
$e$ of the tree $T$ a \emph{weakening parameter} $t_{e}\in[0,1]$. For any two vertices $i$ and $j$ of the graph 
$G$ we define:
\[
 (C_{T})_{ij}={\inf}\big\{t_{e}\,\big|\,e \text{ in the unique path ${\cal P}^
{T}_{i\leftrightarrow j}$ in $T$ joining $i$ and $j$}\big\}\;,
\]
and the infimum is $1$ if $i=j$. 
We arrange $  (C_{T})_{ij} $ in a (symmetric) $V(G)\times V(G)$ matrix $C_T$. The matrix $C_T$ is a \emph{positive matrix}.
This statement is non trivial and its proof can be found in \cite{BKAR} or \cite{advanced}.

To any real, positive, symmetric $n\times n$ matrix $(C_{ij})_{1\leq i,j\leq n}$ we associate a unitary invariant 
normalized Gau\ss ian 
measure $d\mu_{C}(A)$ on $n$ random $N\times N$  Hermitian matrices $A=(A_{1},\dots,A_{n})$ defined by its covariance: 
\begin{equation}
\int d\mu_{C}(A) \;  A_{i|ab}A_{j|cd}=C_{ij}\,\delta_{ad}\delta_{bc} \;, \qquad \int d\mu_{C}(A)=1 \;,
\label{Gaussian:eq}
\end{equation}
where  $A_{i|ab}$ and $A_{j|cd}$ are the matrix elements of the matrices $A_{i}$ and $A_{j}$. This Gau\ss ian measure can be represented
as a differential operator. Indeed, denoting:
\[ \Tr \left[ \frac{\partial}{\partial A_i} \frac{\partial}{\partial A_j} \right]  
 = \sum_{a,b} \frac{\partial}{\partial A_{i|ab}} \frac{\partial}{\partial A_{j|ba}} \;,
\qquad 
\frac{\partial}{\partial A_{i|ab}}=\frac{1}{2}\bigg(\frac{\partial}{\partial \text{Re}A_{i|ab}}-\mathrm{i}\frac{\partial }{\partial \text{Im}A_{i|ab}}\bigg) \; , 
\]
the Gau\ss ian expectation of any function $F(A_1,\dots A_n)$ is:
\[
\int d\mu_{C}(A) \; F(A_1,\dots A_n) =
\left[ e^{\frac{1}{2} \sum_{ij} C_{ij} 
\Tr \left[ \frac{\partial}{\partial A_i} \frac{\partial}{\partial A_j} \right] } F(A_1,\dots A_n)\right]_{A_i=0} \; .
\]

To every loop edge $e\in L(G,T)$ we associate a parameter $s_{e}\in[0,1]$.
As the loop edges are labeled $1,\dots |L(G,T)|$, we will denote $s_1,\dots s_{|L(G,T)|}$ the parameter
associated to the edge $1, \dots |L(G,T)|$. Note that for the loop edge $e=(i,j)$ the parameter $s_e$ 
and the weakening factor $(C_T)_{ij}$ are completely unrelated.
 
We associate to every LVE graph $(G,T)$ the \emph{amplitude} ${\cal A}_{(G,T)}[J,J^{\dagger},\lambda,N]$
defined as a Gau\ss ian integral over $|V(G)|$ Hermitian matrices $A=(A_{i})_{1\leq i\leq |V(G)|}$ (each one of size $N\times N$):
\begin{align}\label{LVEamplitude}
{\cal A}_{(G,T)}[J,J^{\dagger};\lambda,N]=&\frac{(-\lambda)^{|E(G)|}N^{|V(G)|-|E(G)|}}{|V(G)|!}
\mathop{\int}\limits_{1\geq s_{1}\geq\cdots\geq s_{|L(G,T)|}\geq 0}\,\prod_{e\in L(G,T)}ds_{e} \crcr
& \times \mathop{\int}\limits_{[0,1]} \prod_{e\in E(T)} dt_{e}
  \left( \prod_{e=(i,j)\in L(G,T)}\mathop{\text{inf}}\limits_{e'\in P_{i\leftrightarrow j}^{T}} t_{e'}  \right) \crcr
& \times \int d\mu_{s_{|L(G,T)|}C_{T}}(A)\prod_{f\in F(G)}\Tr\bigg\{\mathop{\prod}\limits_{c\in\partial f}^{\longrightarrow}
\bigg(1-\text{i}\sqrt{\frac{\lambda}{N}}\,A_{i_{c}}\bigg)^{-1}(JJ^{\dagger})^{\eta_{c}}\bigg\} \; ,
\end{align}
where:
\begin{itemize}
\item ${\displaystyle\mathop{\prod}\limits_{c\in\partial f}^{\longrightarrow}}$ is the oriented product around the corners $c$ 
on the boundary $\partial f$ of the face $f$.
\item $i_{c}$ is the label of the vertex the corner  $c$ belongs to.
\item $\eta_{c}=1,0$ depending on whether $c$ is followed by a cilium (1) or not (0).
\end{itemize}

We refer to appendix \ref{LVEexamples:app} for some example of LVE graphs and their amplitudes.

The Gaussian measure $d\mu_{s_{|L(G,T)|}C_{T}}(A) $ can also be written as the differential operator:
\[
\int d\mu_{s_{|L(G,T)|}C_{T}}(A) F(A) = \left[  e^{\frac{ s_{|L(G,T)| } }{2} \sum_{ij}[ \inf_{(k,l)\in P^T_{i\leftrightarrow j} } t_{kl} ]
\Tr \; \left[ \frac{\partial}{\partial A_i} \frac{\partial}{\partial A_j} \right] } \; F(A) \right]_{ A_i=0 } \;.
\]

In the case of a LVE graph which is just a tree, $G=T$, we will use the shorthand notation
${\cal A}_{T}[J,J^{\dagger};\lambda,N] \equiv {\cal A}_{(T,T)}[J,J^{\dagger};\lambda,N] $.
The amplitude simplifies drastically in this case: there are no integrals over the $s$ parameters (and $s_{|L(G,T)|}$ is set to $1$),
the product over infima in the second line is empty (hence set to $1$), and only one trace is obtained (as trees have only one face).

\paragraph{Constructive expansions of the generating function}

Let ${\cal C}$  be the cardioid domain  in the complex plane (see figure \ref{cardioidpic}):
 \begin{equation}{\cal C}=\bigg\{\lambda\in{ \mathbb{C} }
 \quad\text{with}\quad 
 4|\lambda|< \cos^{2}\Big({\frac{\arg\lambda}{2}}\Big)\bigg\} \; ,
 \label{cardioid} 
\end{equation}
where we choose the determination $-\pi<\arg\lambda<\pi$ of the argument (hence the argument has a cut on the negative real axis). 

\begin{figure}[htb]
\begin{center}
\includegraphics[width=5cm]{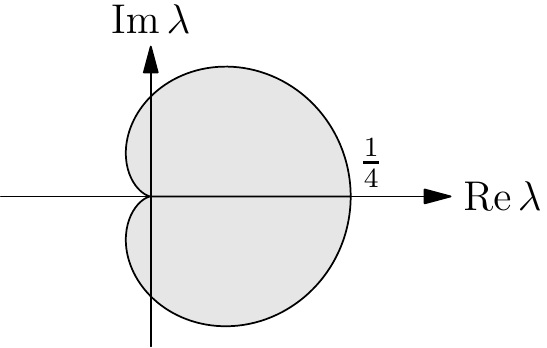}
\caption{analyticity domain in the complex $\lambda$ plane}
\label{cardioidpic}
\end{center}
\end{figure}

Our first result is a convergent expansion of $\log{\cal Z}[J,J^{\dagger};\lambda,N]$ as a sum over LVE trees.
\begin{theorem}[Tree expansion]
\label{treeexpansion}
For any $\lambda\in {\cal C}$, there exists $\epsilon_{\lambda}>0$ depending on $\lambda$ such that for $\|JJ^{\dagger}\|<\epsilon_{\lambda}$ the 
logarithm of ${\cal Z}[J,J^{\dagger};\lambda,N]$ is given by the following absolutely convergent expansion:
\begin{equation}
\log{\cal Z}[J,J^{\dagger};\lambda,N]=\sum_{T
\,\text{LVE tree}}
{\cal A}_{T}[J,J^{\dagger};\lambda,N] \;. 
\label{treeexpansion:eq}
\end{equation}
\end{theorem}

In order to compare the tree expansion of Theorem \ref{treeexpansion} with the conventional perturbative expansion, it is 
necessary to further expand some of the loop edges. The following theorem is obtained by recursively adding loop edges to the 
LVE trees.

\begin{theorem}[Perturbative expansion with remainder]
\label{perturbativegenerating:thm}
For any $\lambda \in {\cal C}$, there exists $\epsilon_{\lambda}>0$ depending on $\lambda$ such that for $\|JJ^{\dagger}\|<\epsilon_{\lambda}$:
\begin{align}
\log{\cal Z}[J,J^{\dagger};\lambda,N]   = &
\sum_{G\text{ ciliated ribbon graph}\atop |E(G)|\leq n}
\frac{(-\lambda)^{|E(G)|}N^{\chi(G)}}{|\text{Aut}(G)|}\prod_{f\in B(G)} \Tr \Big[\big(JJ^{\dagger}\big)^{c(f)}\Big] \crcr
& +  {\cal R}_{n}[J,J^{\dagger};\lambda,N] \; ,
\label{perturbativeexpansion:eq}
\end{align}
where $c(f)$is the number of cilia in the broken face $f$ and the perturbative remainder at order $n$
is a convergent sum over LVE  graphs with at least $n+1$ edges and at most $n+1$ loop edges
\begin{equation}
{\cal R}_{n}[J,J^{\dagger};\lambda,N]=
\sum_{(G,T)\text{ LVE graph}\atop |E(G)|= n+1}
{\cal A}_{(G,T)}[J,J^{\dagger};\lambda,N]
+
\sum_{T\text{ LVE tree}\atop |E(T)|\geq n+2}
{\cal A}_{T}[J,J^{\dagger};\lambda,N] \; .
\end{equation}
\end{theorem}
Note that the first term in \eqref{perturbativeexpansion:eq} involves a sum over ribbon graphs (not LVE graphs) and reproduces the perturbative expansion over maps. 
The remainder is made of the more involved LVE graphs and its amplitude involves further non trivial Gau\ss ian integrations.
In particular, the ribbon graphs in the perturbative expansion do not carry labels on their vertices: this is the origin of the
factor $\frac{1}{|\text{Aut}(G)|}$ (where $\text{Aut}(G)$ is the cardinal of the group of permutations of the labels on the vertices that 
preserves the adjacency relations of the graph). Alternatively, one could also work with labeled ribbon graphs and divide by $|V(G)|!$. 

Since we are dealing with random matrices of size  $N$, it is also possible to organize the expansion in powers of $\frac{1}{N}$. Such 
an expansion is governed by the genus of the ribbon graphs, as a graph with Euler characteristic $\chi(G)$ scales like $N^{\chi(G)}$. 
Contrary to the standard perturbative expansion, the expansion over graphs of fixed genus $g$ has a finite ($\frac{1}{12}$) radius of convergence,
as can be easily seen from their asymptotic behavior \cite{SchaefferChapuyMarcus}. In particular, $\lambda_{c}=-\frac{1}{12}$ is the critical point, 
instrumental in constructing the double scaling limit.
This motivates the introduction of the following cardioid  (see figure \ref{cardioidtopological:pic}):
 \begin{equation}\widetilde{{\cal C}}=\bigg\{\lambda\in{ \mathbb{ C} }
 \quad\text{with}\quad 
 12|\lambda|< \cos^{2}\Big({\frac{\arg\lambda}{2}}\Big)\bigg\} \; .
 \label{cardioidtopological}
\end{equation}
The shift from the factor $4$ for ${\cal C}$ to $12$ for $\widetilde{{\cal C}} $ reflects the fact that the radius of convergence of 
the sum over ribbon graphs of fixed genus is $\frac{1}{12}$ while the radius of convergence of 
the sum over trees is $\frac{1}{4}$. 

\begin{figure}[htb]
\centerline{\includegraphics[width=7cm]{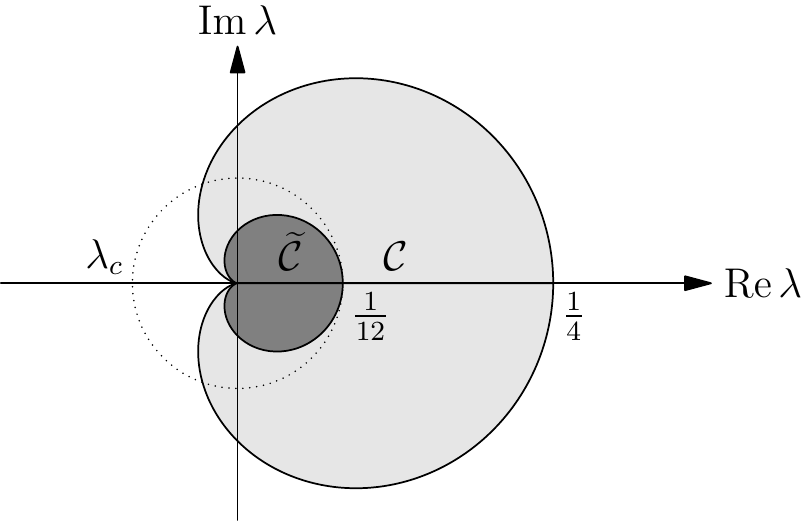}}
 \label{cardioidtopological:pic}
\caption{Analyticity domain of the topological expansion}
\end{figure}

\begin{theorem}[Topological expansion with remainder]
\label{topologicalexpansion:thm}
For any $\lambda \in \widetilde{{\cal C}}$, there exists $\epsilon_{\lambda}>0$ depending on $\lambda$ such that for $\|JJ^{\dagger}\|<\epsilon_{\lambda}$
the logarithm of ${\cal Z}[J,J^{\dagger};\lambda,N]$ is:
\begin{align}
\log{\cal Z}[J,J^{\dagger};\lambda,N]= &
\Bigg(
\sum_{G\text{ ciliated ribbon graph}\atop g(G)\leq g}
\frac{(-\lambda)^{|E(G)|}N^{2-2g(G)-|B(G)|}}{|\text{Aut}(G)|}\prod_{f\in B(G)}\Tr\Big[\big(JJ^{\dagger}\big)^{c(f)}\Big]\Bigg) \crcr
 & + \widetilde{{\cal R}}_{g}[J,J^{\dagger};\lambda,N] \; ,
\label{topologicalexpansion:eq}
\end{align}
where the order $g$ topological remainder is given by the following absolutely convergent expansion:
\begin{equation}
\widetilde{{\cal R}}_{g}[J,J^{\dagger};\lambda,N]=
\sum_{(G,T)\,\text{LVE graph with}\atop\text{$g(G)=g+1$ and $g(G-e_{L(G,T)})=g$}}
{\cal A}_{(G,T)}[J,J^{\dagger},\lambda,N] \; ,
\end{equation}
where $G-e_{L(G,T)}$ is the graph obtained by removing the loop edge with the highest label.
\end{theorem}

This expansion is also obtained by a recursive addition of loop edges to a tree, but with a stop rule which takes into account the topology:
one iteratively adds loop edges as long as the genus of the LVE graph $(G,T)$ does not exceed $g$. 

\paragraph{Cumulants.}
The main objects of interest in this paper are the cumulants (connected correlation functions). 

\begin{definition}[Cumulants]
\label{cumulantsdef}
The \emph{cumulant of order $2k$} is the derivative:
\begin{equation}
{{K}}_{a_{1}b_{1}c_{1}d_{1},\dots,a_{k}b_{k}c_{k}d_{k}}(\lambda,N)=
\frac{\partial^{2}}{\partial J^{\ast}_{a_{1}b_{1}}\partial J^{}_{c_{1}d_{1}}}\cdots
\frac{\partial^{2}}{\partial J^{\ast}_{a_{k}b_{k}}\partial J^{}_{c_{k}d_{k}}}\log{\cal Z}[J,J^{\dagger};\lambda,N]\bigg|_{J=J^{\dagger}=0} \;. 
\end{equation}
\end{definition}
Here $J_{ab}^{\ast}$ is the complex conjugate of $J_{ab}$, so that $(J^{\dagger})_{ab}=J^{\ast}_{ba}$. Note that all the derivatives 
of $\log{\cal Z}$ which are not of this form vanish. For example, the order 2 cumulant is:
\begin{equation}
K_{abcd}(\lambda,N)=\frac{\partial^{2}}{\partial J^{\ast}_{ab}\partial J^{}_{cd}}\log{\cal Z}[J,J^{\dagger}]
=N\Big(
\langle M_{ab}M^{\ast}_{cd}\rangle-\langle M_{ab}\rangle\langle M^{\ast}_{cd}\rangle\Big) \; .
\end{equation}

The normalization is chosen in such a way that the contribution of a genus $g$ graph with $b$ broken faces scales as $N^{2-2g-b}$, corresponding 
to the Euler characteristic of a surface with punctures.

Due to the unitary invariance of the matrix model, the cumulants have a specific form. For any permutation of $k$ 
elements $\sigma\in{\mathfrak S}_{k}$, let us write $C(\sigma)$ the integer partition of $k$ associated to the cycle decomposition 
of $\sigma$ and $|C(\sigma)|$ the number of cycles it contains. Let us also denote by $\Pi_{k}$ the set of integer partitions of $k$ 
(recall that a partition $\pi\in\Pi_{k}$  is an increasing sequence of $|\pi|$ integers $0<k_{1}\leq \cdots\leq k_{|\pi|}$ such 
that $k_{1}+ \cdots + k_{|\pi|}=k$).
To any integer partition of $k$ we associate a \emph{trace invariant}:
\begin{equation}
\Tr_{\pi}(X)=\Tr(X^{k_{1}})\cdots\Tr(X^{k_{p}}) \; .
\end{equation}

As we will see below, the cumulants write in terms of the Weingarten functions $\text{Wg}(\tau\sigma^{-1},N)$ \cite{Collins,ColSni}. These functions
arise when integrating over unitary matrices $\text{U}(N)$ with the invariant normalized Haar measure. 
Denoting $U_{ab}^*$ the complex conjugate of $U_{ab}$ we have \cite{Collins}:
\begin{multline}
\int dU \; U_{a_{1}b_{1}}\dots U_{a_{k}b_{k}}
U^{*}_{c_{1}d_{1}}\dots U^{*}_{c_{l}d_{l}}
=\\
 \delta_{kl}\sum_{\sigma,\tau\in \mathfrak{S}_{k}}
\delta_{a_{\tau(1)c_{1}}}\dots \delta_{a_{\tau(k)}c_{k}}
\delta_{b_{\sigma(1)}d_{1}}\dots \delta_{b_{\sigma(k)}d_{k}}
\text{Wg}(\tau\sigma^{-1},N) \; .
\label{Weingartenrelation}
\end{multline}
The functions $\text{Wg}(\sigma,N)$ only depends on the cycle structure of $\sigma$. 
For low values of $n$, the Weingarten functions read:
\begin{align*}
\text{Wg}\big((1),N\big)&=\frac{1}{N}& \text{Wg}\big((1,1,1),N\big)&=\frac{N^{2}-2}{N(N^{2}-1)(N^{2}-4)}\\
\text{Wg}\big((1,1),N\big)&=\frac{-1}{N^{2}-1}& \text{Wg}\big((1,2),N\big)&=\frac{-1}{(N^{2}-1)(N^{2}-4)}\\
\text{Wg}\big((2),N\big)&=\frac{-1}{N(N^{2}-1)}& \text{Wg}\big((3),N\big)&=
\frac{2}{N(N^{2}-1)(N^{2}-4)} \; .
\end{align*}

Let us chose a permutation $\zeta\in{\mathfrak S}_{k}$ whose 
cycle decomposition reproduces the contribution 
of the broken faces to the amplitude of a LVE graph. Specifically, if there are $b=|B(G)|$ broken faces with $k_{1},\dots,k_{b}$ cilia,
we choose $\zeta$ to have a cycle decomposition of the form:
\begin{equation}
\zeta=(i_{1}^{1}\dots i_{k_{1}}^{1}) \cdots (i_{1}^{b}\dots i_{k_{b}}^{b}) \; .
\end{equation}
This permutation defines a labeling of the cilia in such a way that the product of traces over the broken faces can be expressed as:
\begin{equation}
\prod_{1\leq m\leq b}\Tr\Big[JJ^{\dagger}
\mathop{\prod}\limits_{1\leq r\leq k_{m}}^{\longrightarrow}X^{i^{m}_{r}}\Big]=
\sum_{1\leq p_{1}, q_1 \dots \leq N} \prod_{1\leq l\leq k} (JJ^{\dagger})_{p_{l}q_{l}}  X^{l}_{q_{l}p_{\zeta(l)}} \;, 
\end{equation}
where $X^{l}$ is the product of the resolvents located on the corners separating the cilia labeled $l$ and $\zeta(l)$. Similarly, 
for the $F(G)-B(G)$ unbroken faces we denote by $Y^{m}$ the product of the resolvents around the unbroken face labeled $m$. 

\begin{proposition}\label{prop:ampliWeingarten}
 The amplitude of a LVE graph in eq. \eqref{LVEamplitude} expands in trace invariants as:
 \begin{equation} \label{amplitudetraceinvariant}
 {\cal A}_{(G,T)}[J,J^{\dagger},\lambda,N]  =\sum_{\pi\in\Pi_{k}}{A}_{(G,T)}^{\pi}(\lambda,N) \; \Tr_{\pi}(JJ^{\dagger}) \; ,
  \end{equation}
with
\begin{multline}
{A}_{(G,T)}^{\pi}(\lambda,N)=\frac{(-\lambda)^{|E(G)|}N^{|V(G)|-|E(G)|}}{|V(G)|!}
\mathop{\int}\limits_{1\geq s_{1}\geq\cdots\geq s_{|L(G,T)|}\geq 0}\,\prod_{e\in L(G,T)}ds_{e}\\
\int \prod_{e\in E(T)} dt_{e}
 \left( \prod_{e=(i,j) \in L(G,T)}\mathop{\text{inf}}\limits_{e'\in P_{i\leftrightarrow j}^{T}} t_{e'} \right) 
\int d\mu_{s_{|L(G,T)|}C_{T}}(A) \\
\times \sum_{\tau,\sigma\in{\mathfrak S}_{k}\atop C(\sigma)=\pi}
\sum_{1\leq p_{1},\dots,p_{k}\leq N}\text{Wg}(\tau\sigma^{-1},N)
\prod_{1\leq m\leq F(G)-B(G)}\Tr\Big[Y^{m}\Big]\prod_{1\leq l\leq k}
X^{l}_{p_{\tau(l) }p_{\zeta(l)}}\label{amplitudeWg} \; .
\end{multline}
\end{proposition}
 
If the LVE graph $(G,T)$ is reduced to a tree we use the shorthand notation ${A}_{T}^{\pi}(\lambda,N) $ instead of 
${A}_{(T,T)}^{\pi}(\lambda,N)$. 

\begin{proposition}[Scalar cumulants]
\label{structure:prop}
The order $2k$ cumulants can be written as a sum over  partitions of $k$ and over two permutations of $k$ elements:
\begin{equation}
\label{structure:eq}
{{K}}_{a_{1}b_{1}c_{1}d_{1},\dots,a_{k}b_{k}c_{k}d_{k}}(\lambda,N)=\sum_{\pi\in\Pi_{k} }K_{\pi}(\lambda,N)
\sum_{\rho,\sigma\in\mathfrak{S}_{k}}
\prod_{1\leq l\leq k}\delta_{c_{l},a_{\rho\tau_{\pi}\sigma^{-1}(l)}}\delta_{d_{l},b_{\rho\xi_{\pi}\sigma^{-1}(l)}} \; ,
\end{equation}
where $\tau_{\pi}$ and $\xi_{\pi}$ are arbitrary permutations such that $\tau_{\pi}(\xi_{\pi})^{-1}$ has a cycle structure corresponding to the partition $\pi$
and the \emph{scalar cumulants} $ K_{\pi}(\lambda,N) $ are given by the expansion:
\begin{equation}
K_{\pi}(\lambda,N)=\sum_{T \text{ LVE tree with $k$ cilia}}{\cal A}^{\pi}_{T}(\lambda,N)\label{Kpitree} \; .
\end{equation}
\end{proposition}

Choosing any other pair of permutations $\tau_{\pi}$ and $\xi_{\pi}$ leads to an identical result, after reorganizing the sum over 
$\rho$ and $\sigma$. $K_{\pi}(\lambda,N)$ only depends on the partition $\pi$ and not on the index structure of 
${{K}}_{a_{1}b_{1}c_{1}d_{1},\dots,a_{k}b_{k}c_{k}d_{k}}(\lambda,N)$ which explains why we call it \emph{scalar cumulant}.   

The main goal of this paper is to establish some analyticity results as well as bounds for the scalar cumulants $K_{\pi}(\lambda, N)$ regarded as 
functions of $\lambda$ inside  a cardioid  with $N$ considered as a parameter.
 
 \paragraph{Constructive expansions for cumulants}

Our first result states that the expansion of  $K_{\pi}(\lambda, N)$ as a sum over trees
yields an analytic function of $\lambda\in {\cal C}$. 

 \begin{theorem}[Analyticity and bound for cumulants] \label{treecumulants:thm}
The series:
\begin{equation}
K_{\pi}(\lambda, N)=\sum_{T\text{ LVE tree with $k$ cilia}}{\cal A}^{\pi}_{T}(\lambda, N)\;,
\end{equation}
defines an analytic function of $\lambda\in{\cal C}$. Moreover, each term in this sum is bounded (for $N$ large enough) as:
\begin{equation}\label{treecumulantsbound}
\big|{\cal A}^{\pi}_{T}(\lambda, N)\big|\leq\frac{N^{2-|\pi|}|\lambda|^{|E(T)|}\,(k!)^2 \, 2^{2k}}{(\cos\frac{\arg\lambda}{2})^{2|E(T)|+k}\,|V(T)|!}
\; ,
\end{equation}
where $|\pi|$ is the number of integers in the partition $\pi$ of $k$ (number of cilia).
\end{theorem}

By further expanding loop edges on each tree, we obtain a perturbative expansion with a well controlled remainder. 
In order to identify the graphs contributing to $K_{\pi}(\lambda,N)$, we say that  a ciliated ribbon graph has broken 
faces corresponding to  $\pi$ if the partition of the cilia defined by the broke faces agrees with the partition $\pi$. 

\begin{theorem}[Perturbative expansion with remainder] \label{perturbativecumulants:thm}
The perturbative expansion of the cumulants reads:
\begin{equation}
K_{\pi}(\lambda, N)=\sum_{G\text{ ribbon graph with $k$ cilia}\atop\text{ broken faces corresponding to $\pi$ and $|E(G)|\leq n$}}
\frac{(-\lambda)^{|E(G)|}N^{\chi(G)}}{|\text{Aut}(G)|}
+{\cal R}_{\pi,n}(\lambda,N) \; .
\end{equation}
The perturbative remainder  ${\cal R}_{\pi,n}(\lambda,N)$ is a sum over LVE graphs with  $k$ cilia, at least $n+1$ edges and at most $n+1$ loop edges,
\begin{equation}
{\cal R}_{\pi,n}(\lambda,N)=
\sum_{(G,T)\text{ LVE graphs with broken structure corresponding to $\pi$}\atop
|E(G)|\geq n+1 \text{ and } |L(G,T)|\leq n+1}
{\cal A}^{\pi}_{(G,T)}(\lambda, N) \; .
\end{equation}
The perturbative reminder is analytic for $\lambda \in {\cal C}$ and for any $\lambda \in {\cal C}$ and $N$ large enough 
it obeys the bound:
\begin{align*}
& \Big|{\cal R}_{\pi,n}(\lambda,N)\Big|
\leq \crcr
& \; \leq N^{2-|\pi|} \left( \frac{2^{3k-1}k!}{ \big(\cos\frac{\arg\lambda}{2}\big)^{k} } \right) (n+1)!   
\left(  \frac{ 4 |\lambda|}{  \big(\cos\frac{\arg\lambda}{2}\big)^{2}} \right)^{n+1} 
\left(
 \frac{    \frac{ 4 |\lambda|}{  \big(\cos\frac{\arg\lambda}{2}\big)^{2} }   }
   { \left( 1 -  \frac{ 4 |\lambda|}{  \big(\cos\frac{\arg\lambda}{2}\big)^{2} } \right)^{n+2} }
   + 2^{k+n+2 }
\right)\; .
\end{align*}
\end{theorem}

\paragraph{Borel summation for cumulants}

The previous expansion defines an asymptotic expansion of the cumulants. Indeed, let us collect the contribution of all graphs of a given order in
\begin{equation}
a_{\pi,n}(N)=\sum_{G\text{ ribbon graph with $k$ cilia}\atop\text{ broken faces $\pi$  and $|E(G)|=n$}}\frac{N^{\chi(G)}}{|\text{Aut}(G)|} \; ,
\end{equation}
so that for $\lambda \in {\cal C}$ the bound on $R_{\pi,n}(\lambda,N)$ implies 
\begin{equation}
\lim_{\lambda\rightarrow 0}\Bigg|
\frac{K_{\pi}(\lambda, N)-\sum_{k\leq m\leq n}(-\lambda)^{m}a_{\pi,m}(N)}
{\lambda^{n}}
\Bigg|=0 \; .
\end{equation}

However, the series $\sum_{n}a_{\pi,n}\lambda^{n}$ is divergent which means that  $K_{\pi}(\lambda,N)$ is not analytic at the origin.
From a combinatorial point of view, the divergence of the series is due to the occurrence of too many graphs at a given order in $n$.
Nevertheless,  $\sum_{}a_{\pi,n}(N)\lambda^{n}$ contains all the information required to reconstruct  $K_{\pi}(\lambda,N)$ through 
the Borel summation procedure. The latter is based on the following theorem.

For any $R>0$, let ${\cal D}_{R}$ be the disc of radius $R$ tangent at the origin (see figure \ref{Borel:pic})
 \begin{equation}
 {\cal D}_{R}=\big\{\lambda\in{ \mathbb{C} }\,\big|\,\text{Re}\Big(\frac{1}{\lambda}\Big)>\frac{1}{R}\Big\} \; ,
 \end{equation}
and let $\Sigma_{\sigma}$ be the half strip (see figure \ref{Borel:pic} for a representation of ${\cal D}_{R}$ and $\Sigma_{R}$)
\begin{equation}
\Sigma_{\sigma}=\big\{s\in{ \mathbb{C} }\,\big|\,\text{distance}(s,{ \mathbb{R}^{+}})<\frac{1}{\sigma}\} \; .
\end{equation} 

\begin{figure}[htb]
\[
\begin{array}{cc}
\includegraphics[width=4cm]{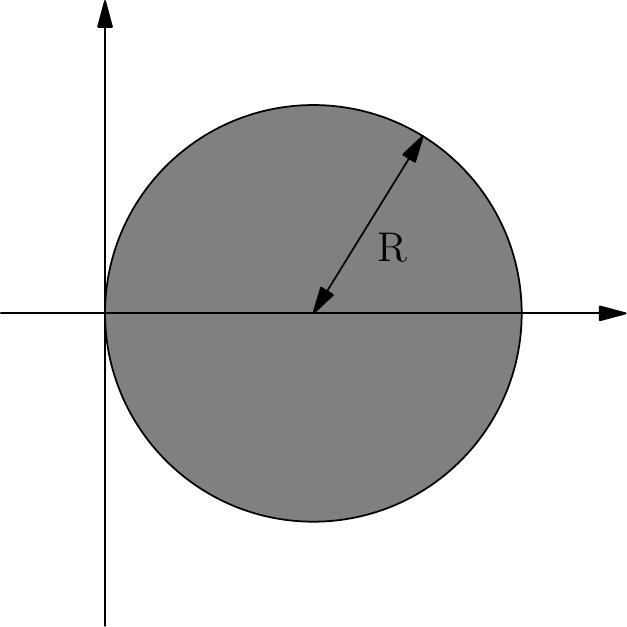}&
\includegraphics[width=6.7cm]{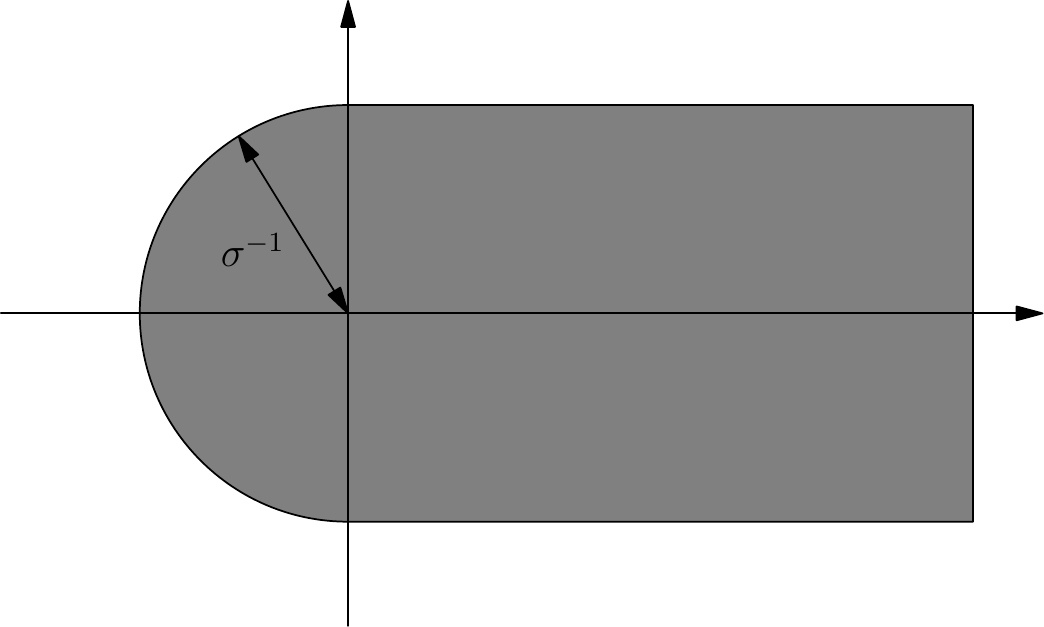}\\
{\cal D}_{R}&\Sigma_{\sigma}
\end{array}
\]
\caption{Domain of analyticity of $F$ and of its Borel transform $B$}
\label{Borel:pic}
\end{figure}

\begin{theorem}[Nevanlinna-Sokal \cite{Sokal}]
\label{NevanlinnaSokal} Let $R>0$ and $F_{\omega}(\lambda)$ be a family of analytic functions on the disc ${\cal D}_R$ depending on some 
parameter $\omega\in\Omega$. If there exists a sequence $a_{n}(\omega)$ of  functions of $\omega\in\Omega$ obeying, 
for any $n$, $\lambda\in{\cal D}_{R}$ and $\omega\in\Omega$ the uniform bound: 
  \begin{equation}
 \big| F_{\omega}(\lambda)-\sum_{m=0}^{n}a_{m}(\omega)\lambda^{m} \big|< C\sigma^{n+1}|\lambda|^{n+1}(n+1)!\label{Taylorbound} \; ,
 \end{equation}
 with $C$ and $\sigma$ two positive constants that do not depend on $\omega$, then the series
 \begin{equation}
B_{\omega}(s)=\sum_{n=0}^{\infty}\frac{a_{n}(\omega)}{n!}s^{n} \;,
 \end{equation}
 has radius of convergence $\sigma^{-1}$ and can be analytically continued in the strip $\Sigma_{\sigma}$.
Moreover, there exists a constant $B$ such that, for any $s\in\Sigma_{\sigma}$ and $\omega\in\Omega$, we have
\begin{equation}
\big|B_{\omega}(s)\big|\leq B \text{e}^{\frac{s}{R}} \;.
\end{equation}
Finally, for any $\lambda\in{\cal D}_{R}$, $F_{\omega}(\lambda)$ is given by the following absolutely convergent integral:
\begin{equation}
F_{\omega}(\lambda)=\int_{0}^{\infty}\!\!ds \,B_{\omega}(s)
\text{e}^{-\frac{s}{\lambda}}\label{Boreltransform} \; .
\end{equation}
\end{theorem}
If the assumption of theorem \ref{NevanlinnaSokal} are fulfilled, $F_{\omega}$ is said to be Borel summable at $\lambda=0$, uniformly in $\omega$. 
In this case, $F_{\omega}$ can be uniquely recovered from the coefficients $a_{n}(\omega)$ 
using its Borel transform $s\mapsto B_{\omega}(s)$ and eq. \eqref{Boreltransform}.

For any $\lambda\in {\cal D}_R \subset {\cal C}$, $\big(\cos\frac{\arg\lambda}{2}\big) \ge 2^{-1/2}$ and there exists $R$
such that the perturbative reminder in Theorem \ref{perturbativecumulants:thm} is bounded as in eq. \eqref{Taylorbound}.

\begin{corollary}[Borel summability]
\label{Boreltheorem}
The rescaled cumulants  $N^{-2+|\pi|}K_{\pi}(\lambda, N)$ (with $|\pi|$ the number of parts in the partition $\pi$) 
are Borel summable in $\lambda$ at the origin, uniformly in $N$, so that
\begin{equation}
K_{\pi}(\lambda, N)=\int_{0}^{\infty}\!\!\!\!ds\,\text{e}^{-\frac{s}{\lambda}}
\bigg(\sum_{n\geq k}\frac{a_{\pi,n}(N)}{n!}s^{n}\bigg) \;, \end{equation}
in a disc included in ${\cal C}$ tangent to the imaginary axis at the origin and independent of $N$.
\end{corollary}

\paragraph{Topological expansion for cumulants}

The Taylor expansion at the origin of the cumulants leads to ribbon graphs drawn on surfaces with boundary. 
The Euler characteristic  of a surface determines the power of $N$. This is known as the topological expansion. 
While it is well known that the contributions of Feynman graphs of fixed genus are analytic functions in a disk of
fixed radius $\frac{1}{12}$, less is known about the remainder. We state an analyticity results and a bound for 
the remainder. 

\begin{theorem}[Topological expansion]
\label{topologicalcumulants:thm}
The cumulants $K_{\pi}(\lambda, N)$ are expanded in inverse powers of $N$ as
\begin{equation}
K_{\pi}(\lambda, N)=\sum_{h=0}^{g} N^{2-2g-|\pi|}K_{\pi,h}(\lambda)+\widetilde{R}_{\pi,g}(\lambda,N) \;, 
\end{equation}
where $K_{\pi,h}(\lambda)$ is a sum over ciliated ribbon graphs of genus $h$ whose broken faces correspond to the partition $\pi$, 
convergent for $|\lambda|<\frac{1}{12}$:
\begin{equation}
K_{\pi,h}(\lambda)=\sum_{G\text{ ribbon graph with}\atop\text{genus $h$ and broken faces corresponding to $\pi$}} \frac{(-\lambda)^{|E(G)|}}{|\text{Aut}\,G|} \; .
\end{equation}
The topological remainder $\widetilde{R}_{\pi,g}(\lambda,N)$ is a sum over LVE 
graphs with broken faces corresponding to $\pi$, genus $g+1$ and such that, if we remove the loop edge of highest label, we get a genus $g$ graph 
\begin{equation}
\widetilde{{\cal R}}_{\pi,g}(\lambda,N)=
\sum_{(G,T)\text{ LVE graphs with broken faces corresponding to $\pi$}\atop
g(G)=g+1 \text{ and } g(G-e_{|L(G,T)|})=g}
{\cal A}^{\pi}_{(G,T)}(\lambda, N) \; .
\end{equation}
This series converges for $\lambda\in\widetilde{\cal C}$ and in this domain the topological reminder is bounded by 
\begin{align*}
 & \big|\widetilde{R}_{\pi,g}(\lambda,N)\big| \leq \crcr
 &\leq  N^{2-2(g+1)-|\pi|} \frac{2^{3k} k! }{ \big(\cos\frac{\arg\lambda}{2}\big)^{k} } C''_{g+1}
\left( \frac{12|\lambda|}{\big(\cos\frac{\arg\lambda}{2}\big)^{2} } \right)^{2g+2} 
\frac{(4g+k+1)! }{ \left( 1 -    \frac{12|\lambda|}{\big(\cos\frac{\arg\lambda}{2}\big)^{2} } \right)^{4g+k} } \;, 
\end{align*}
 with $C''_g$ a constant depending only on the genus.
\end{theorem}

\section{Intermediate field representation}\label{sec:intfield}

\label{intermediate}

To begin with, we introduce the intermediate field $A$ (a $N\times N$ Hermitian matrix) and write the quartic interaction as a Gau\ss ian integral:
\begin{equation}
\exp \left\{ -\frac{\lambda }{2N}\Tr(MM^{\dagger}MM^{\dagger}) \right\}
=\int dA \exp \left\{ -  \frac{1}{2}\Tr(A^{2}) + \mathrm{i}\sqrt{\frac{\lambda}{N}}\,\Tr(M^{\dagger}AM)  \right\} \;, 
\label{intermediateA}
\end{equation}
where the integral is over Hermitian $N\times N$ matrices and is assumed to be normalized. The new field $A$ propagates with the trivial Gau\ss ian
measure and the four valent interaction is traded for a three valent interaction involving an $A$ field and a $M$ and a $M^{\dagger}$ field.
This is illustrated in figure \ref{intermediate:fig}.

\begin{figure}[htb]
\[
\parbox{3cm}{\includegraphics[width=3cm]{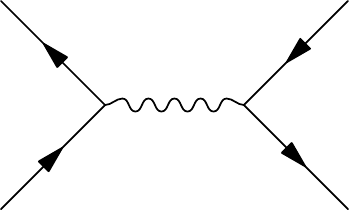}}\qquad
\Leftrightarrow\qquad\
\parbox{2cm}{\includegraphics[width=2cm]{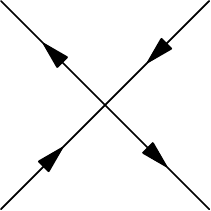}}
\]
\caption{Intermediate field representation.}
\label{intermediate:fig}.
\end{figure}

The generating function is thus:
\begin{multline}
{\cal Z}[J,J^{\dagger};\lambda,N]=\int dMdA\\
\exp\bigg\{-\frac{1}{2}\Tr(A^{2}) - 
\Tr\bigg[ M^{\dagger}\bigg(1-\mathrm{i}\sqrt{\frac{\lambda}{N}}A\bigg)M\bigg]
+\sqrt{N}\Tr(JM^{\dagger})+\sqrt{N}\Tr(MJ^{\dagger})\bigg\} \; ,
\label{AMintegral}
\end{multline}
The integral over the original matrices $M$ and $M^{\dagger}$ is a (non normalized) Gau\ss ian integral with covariance
$\big(1-\mathrm{i}\sqrt{\frac{\lambda}{N}}A\big)\otimes 1$. Taking into account that:
\begin{equation}
\det \left[ \Big(1-\mathrm{i}\sqrt{\frac{\lambda}{N}}A \Big)\otimes 1 \right] = 
\bigg[\det\Big(1-\mathrm{i}\sqrt{\frac{\lambda}{N}}A\Big)\bigg]^{N}=
\exp \bigg\{N\Tr\log\Big(1-\mathrm{i}\sqrt{\frac{\lambda}{N}}A\Big)\bigg\} \;,
\end{equation}
we obtain:
\begin{multline}
{\cal Z}[J,J^{\dagger};\lambda,N]=\\
\int dA
\exp\bigg\{ - \frac{1}{2}\Tr(A^{2})
-N\Tr\log\left(1-\mathrm{i}\sqrt{\frac{\lambda}{N}}A\right)-N
\Tr \left[ J\bigg(1-\mathrm{i}\sqrt{\frac{\lambda}{N}}A \bigg)^{-1}J^{\dagger}\right] \bigg\}\label{Aintegral} \; .
\end{multline}

We thus have three different expressions \eqref{Mintegral}, \eqref{AMintegral} and \eqref{Aintegral} 
for the generating function of the cumulants $\log{\cal Z}[J,J^{\dagger};\lambda,N]$. Their Feynman graph expansions are constructed as follows.

The expression \eqref{AMintegral} involves two types of fields $A$ and $M$ so that the Feynman graphs have tow types of edges. The $M$ edges (solid edges) 
are oriented from $M^{\dagger}$ to $M$ since $M$ is a complex matrix while the $A$ edges (wavy edges) are not because $A$ is Hermitian. 
There are 3-valent vertices corresponding to $\Tr(M^{\dagger}AM)$ and univalent vertices, also viewed as extra half-edges 
(external legs in the physics literature) corresponding to $\Tr(JM^{\dagger})$ and $\Tr(MJ^{\dagger})$.
Note that all the variables we integrate over are matrices so that we have a cyclic ordering at each vertex and
the Feynman graphs are ribbon graphs. We embed the trivalent vertices turning in the clockwise direction so that 
the $A$ edges are on the right when we follow the orientation of the $M$ edges.

\begin{figure}[htb]
\begin{center}
\includegraphics[width=3cm]{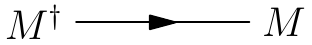}\qquad\qquad\qquad
\includegraphics[width=3cm]{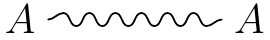}\qquad\\
\bigskip
\bigskip
\includegraphics[width=3.5cm]{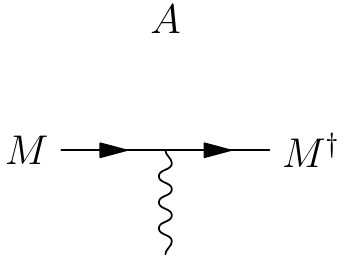}\qquad\qquad
\includegraphics[width=2.5cm]{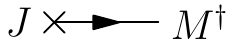}\qquad\qquad
\includegraphics[width=2.5cm]{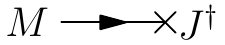}\qquad
\end{center}
\caption{Propagators and interaction}
\end{figure}

Integrating over $A$  in \eqref{AMintegral} before proceeding to the perturbative expansion, we recover the integral \eqref{Mintegral}. 
Its Feynman rules involve only the $M$ edges (which are oriented) as well as an even number of univalent vertices (external legs) 
and tetravalent vertices. The latter involve two incoming edges and two outgoing ones, 
alternating in cyclic order around the vertex. 

Integrating over $M$ in \eqref{AMintegral} before proceeding to the perturbative expansion, we recover the integral \eqref{Aintegral}. 
Its Feynman rules involve the $A$ edges and two types of vertices of arbitrary valence (see figure \ref{comparison:fig}). 
\begin{figure}[htb]
\[
\parbox{2cm}{\includegraphics[width=2cm]{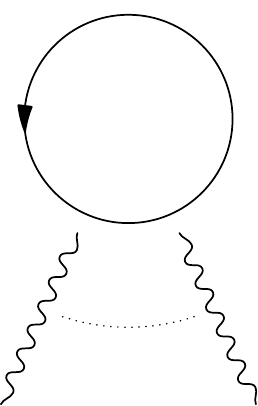}}\qquad
\Leftrightarrow\qquad\
\parbox{2cm}{\includegraphics[width=2cm]{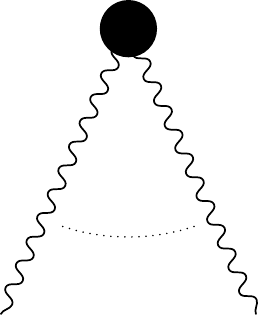}}
\qquad\qquad
\parbox{2cm}{\includegraphics[width=2cm]{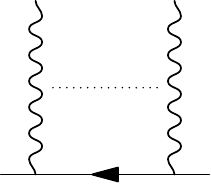}}\qquad
\Leftrightarrow\qquad\
\parbox{2cm}{\includegraphics[width=2cm]{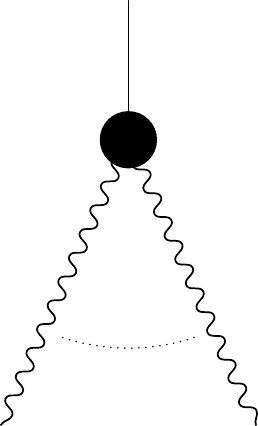}}
\]
\caption{Intermediate field vertices.}
\label{comparison:fig}
\end{figure}
The first one is an ordinary ribbon vertex, arising from the term $\Tr\log\big(1-\mathrm{i}\sqrt{\frac{\lambda}{N}}A\big)$. The second one comes 
from the coupling to the source $\Tr J\big(1-\mathrm{i}\sqrt{\frac{\lambda}{N}}A\big)^{-1}J^{\dagger}$. It is a ribbon vertex with a cilium on a corner 
(the insertion of the source). We illustrate the three representations for a graph contributing to the order 2 cumulant in figure \ref{graphexample:fig}.
\begin{figure}[htb]
\[
\begin{array}{ccccc}
\parbox{6cm}{\includegraphics[width=6cm]{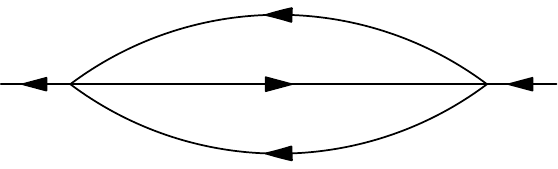}}&
\Leftrightarrow&
\parbox{6cm}{\includegraphics[width=6cm]{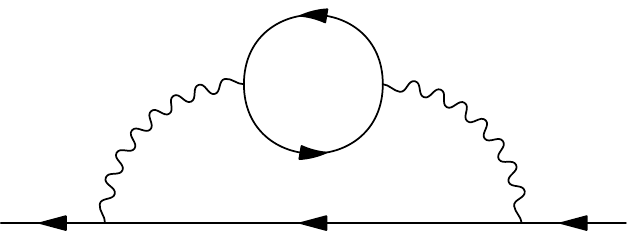}}&
\Leftrightarrow&
\parbox{1.5cm}{\includegraphics[width=1.5cm]{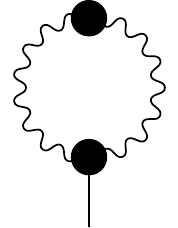}}\\
&&&&\\
\text{ variable $M$}&&
\text{ variables $M$ and $A$}&&
\text{ variable $A$}
\end{array}
\]
\caption{Three equivalent graphs contributing to the order 2 cumulant.}
\label{graphexample:fig}
\end{figure}

The perturbative expansion can be performed either starting from \eqref{Mintegral} or starting from \eqref{Aintegral}, using:
\begin{equation}
\bigg(1-\mathrm{i}\sqrt{\frac{\lambda}{N}}A\bigg)^{-1}=\sum_{n=0}^{\infty}
\bigg(\mathrm{i}\sqrt{\frac{\lambda}{N}}\bigg)^{n}A^{n}
\end{equation}
and performing the Gau\ss ian integral over $A$.

Comparing the two perturbative expansion for the order 2 cumulant yields the following bijection.
\begin{proposition}
\label{bijection}
The intermediate field representation yields the following  bijection
\[
\left\{\text{
\begin{minipage}{4cm}
Connected alternating 2-in 2-out ribbon graphs with $n$ vertices and $2m$ external edges
\end{minipage}
}
\right\}
\quad\Leftrightarrow\quad
\left\{\text{
\begin{minipage}{4cm}
Connected ribbon graphs with $n$ edges and $m$ ciliated vertices
\end{minipage}
}
\right\}
\]
\end{proposition}

 This bijection can be described explicitly as follows. Starting with a connected alternating 2-in 2-out ribbon graph with $n$ vertices and $2m$ external edges, 
 we observe that its faces come in three types. If the face does not contain an external edge, either it is on the left or it is on the right of all 
 the edges which bound it. We color the first kind of faces in black and the second kind white. If a face is broken, then it contains an even
 number of external edges that are alternatively incoming and outgoing. The pieces of the face comprised between two consecutive external edges are either 
 on the left or on the right of all the edges which bound them. We color these pieces of faces in black and respectively in white.
 We join pairs of incoming and outgoing external edges separated by black pieces of broken faces into cilia. The black faces (ciliated or not) 
 define the vertices of the intermediate field graphs. Two such vertices are joined by an intermediate field edge if and only 
 if the associated faces meet at vertex.

Conversely, given a intermediate field graph, we expand its vertices into (black) faces, and we cut the cilia into two. 
We then form tetravalent vertices by contracting the intermediate field edges. 

This construction is a generalization of the medial graph construction to graphs with external edges (or equivalently, cilia). Indeed, if there
is no cilium on the $A$ graph, then the associated $M$ graph is its medial graph. The basic features of this bijection are summarized in the following  table.

\begin{center}
\begin{tabular}{|c|c|}
\hline
matrix model&intermediate field\\
\hline
vertex&edge\\
\hline
black face & vertex\\
\hline
white face & face\\
\hline 
edge&corner\\
\hline
pair of external legs&cilium\\
\hline
\end{tabular}
\label{correspondence}
\center{Matrix model graphs - intermediate field graphs correspondence.}
\end{center}

Let us end this section by giving two consequences of the intermediate field representation of the matrix model. 

First, the number of planar graphs with $n$ vertices contributing to the order 2 cumulant (2-point function in physics parlance) 
can be evaluated explicitly using the Schwinger-Dyson equation for the intermediate field. The details of this
computation are relegated to the appendix \ref{SDappendix} and the result is:
\begin{equation}
\frac{2\cdot3^{n}}{n+2}C_{n}\qquad\text{with}\quad C_{n}=\frac{(2n)!}{n!^{2}(n+1)}\quad(\text{Catalan numbers})   \; .\label{2point}
\end{equation}
This is nothing but the number of planar bipartite quadrangulations with $n$ quadrangles, rooted at an edge. 
Bipartiteness means that the vertices of the quadrangulation
are colored in black and white and the edges only connect vertices of different colors.  The $M$ graph is 
the dual of the quadrangulation. The black/white coloring of the faces of the $M$ graph induces an orientation of the 
$M$ edges in such a way that all the $M$ vertices are alternating 2-in 2-out. 

The intermediate field graphs are in bijection with bipartite quadrangulations with $m$ marked edges. The intermediate field graph is 
obtained by connecting the pair of black vertices
on each quadrangle by an wavy $A$ edge (and adding a cilium for every incidence of a marked edge at a black vertex). We thus obtain:
\begin{proposition}
\label{dualbijection}
\[
\left\{\text{
\begin{minipage}{4cm}
Bipartite quadrangulations with $n$ faces of genus $g$ having $m$ marked edges
\end{minipage}
}
\right\}
\quad\Leftrightarrow\quad
\left\{\text{
\begin{minipage}{4cm}
Connected ribbon graphs of genus $g$ with $n$ edges and $m$ cilia.
\end{minipage}
}
\right\}
\]
\end{proposition}

 \begin{figure}[htb]
 \begin{center}
 \begin{tabular}{cc}
\parbox{7cm}{\includegraphics[width=7cm]{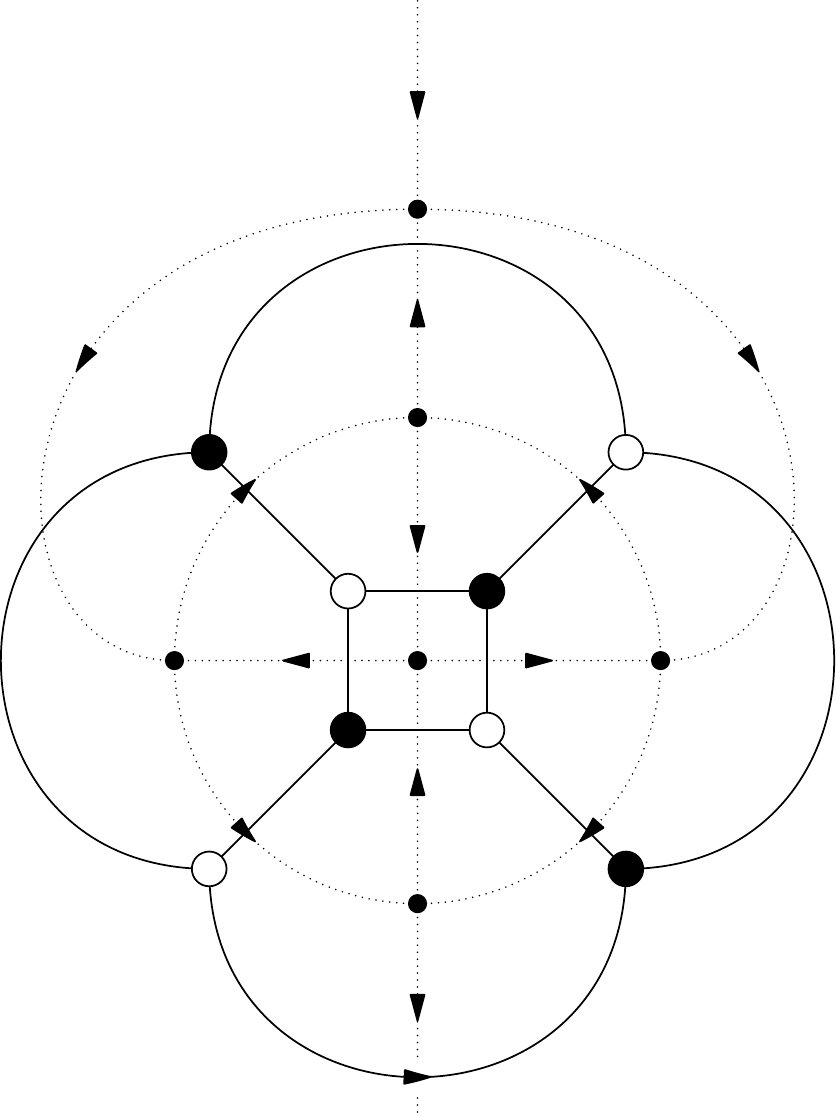}}\qquad&\qquad
\parbox{6cm}{\includegraphics[width=6cm]{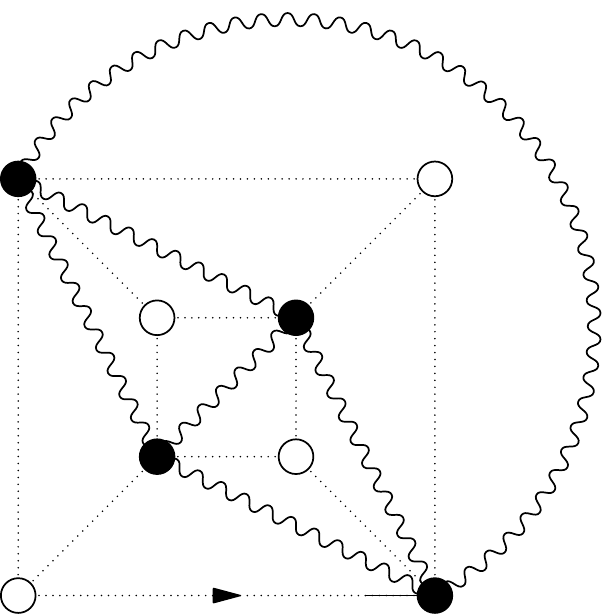}}\\
Dual quadrangulation & Ciliated graph
 \end{tabular}
 \end{center}
\caption{Bijection between bipartite quadrangulations with marked edges and ribbon graphs with cilia.}
 \label{bijectionquadra:fig}
 \end{figure}
 
Second, the intermediate field can be used to study the analyticity properties of ${\cal Z}$.
We first have the following bound.
 \begin{lemma}
 \label{boundresolvent}
Writing $\lambda=\rho\mathrm{e}^{\mathrm{i}\theta}$ with $\rho>0$, we have:
 \begin{equation}
 \Big\|\Big(1-\mathrm{i}\sqrt{\frac{\lambda}{N}}A\Big)^{-1}\Big\|\leq\frac{1}{\cos\frac{\theta}{2}} \; ,
 \end{equation}
 where $ \Big\| \cdot \Big\| $ denotes the operator norm.
 \end{lemma}
 
 \begin{proof}
 To prove the lemma, it is convenient to factor $\sqrt{\lambda}$ and write
\begin{equation}
\Big(1-\mathrm{i}\sqrt{\frac{\lambda}{N}}A\Big)^{-1}
=\frac{1}{\sqrt{\lambda}}
\int_{0}^{\infty}d\alpha\,\exp\Big\{  -\alpha  \frac{1}{\sqrt{\lambda}} + \alpha \frac{\text{i}}{\sqrt{N}} A\Big\} \;,
\end{equation}
therefore, the operator norm is bounded by:
\begin{equation}
\Big\|\Big(1-\mathrm{i}{ \sqrt{\frac{\lambda}{N}}}A\Big)^{-1}
\Big\|
\leq\frac{1}{|\sqrt{\lambda}|}
\int_{0}^{\infty}\exp\Big\{-\alpha\text{Re}\big(\frac{1}{\sqrt{\lambda}}\big)\Big\}
\Big\|\exp\Big\{   \alpha \frac{\text{i}A}{\sqrt{N}}\Big\}\Big\|
=\frac{1}{\cos\frac{\theta}{2}} \; .
\end{equation}
\end{proof}
We can then rewrite \eqref{Aintegral} as:
\begin{equation}
{\cal Z}[J,J^{\dagger};N,\lambda]=
\int dA
\frac{\exp-\bigg\{\frac{1}{2}\Tr (A^{2})
+N\Tr J\bigg(1-\mathrm{i}\sqrt{\frac{\lambda}{N}}A
\bigg)^{-1}J^{\dagger}\bigg\} }{\Big[\det \Big(1-\mathrm{i}\sqrt{\frac{\lambda}{N}}A
\Big)\Big]^{N}} \; ,
\end{equation}
and use lemma \ref{boundresolvent} to show that this integral is convergent for $\theta\in (-\pi,\pi)$. 
As the integrand is analytic for $\theta\in (-\pi,\pi)$ we have the following result.
\begin{proposition}
${\cal Z}[J,J^{\dagger};N,\lambda]$ is analytic in $\lambda$ on the cut plane ${\mathbb C}-{\mathbb R}^{-}$. 
\end{proposition}

However, analyticity of ${\cal Z}[J,J^{\dagger};N,\lambda] $ in the cut plane does not imply analyticity of its logarithm, 
as ${\cal Z}[J,J^{\dagger};N,\lambda]  $ may have zeros. In the next section we will see that in order to establish an 
analyticity result for the logarithm one needs to work some more.

\section{Proofs of the theorems regarding the generating function}\label{sec:proof1}

In this section, we establish the constructive theorems 
\ref{treeexpansion}, \ref{perturbativegenerating:thm} and \ref{topologicalexpansion:thm}
regarding the generating function of the cumulants.

\subsection{The Loop Vertex Expansion (proof of Theorem \ref{treeexpansion})}\label{LVEsec}

The basic ingredient in establishing the constructive theorems stated in section \ref{mainresultsec} is the loop vertex expansion, introduced 
by Rivasseau in \cite{LVE}. Starting with \eqref{Aintegral}, we expand the exponential as a power series, convergent if $\lambda\in{\mathbb C}-{\mathbb R}^{-}$, 
\begin{multline}
{\cal Z}[J,J^{\dagger}]=
\sum_{n=0}^{\infty}\frac{(-1)^{n}}{n!}
\int d\mu(A)\,\,
\bigg[N\Tr\log\bigg(1-\mathrm{i}\sqrt{\frac{\lambda}{N}}A\bigg)+N\Tr J\bigg(1-\mathrm{i}\sqrt{\frac{\lambda}{N}}A
\bigg)^{-1}J^{\dagger}\bigg]^{n} \; ,
\end{multline}
where $d\mu(A)=dA
\exp-\frac{1}{2}\Tr(A^{2})$ is the normalized Gau\ss ian measure on Hermitian matrices (and we dropped the arguments
$\lambda$ and $N$ of ${\cal Z}$ in order to simplify the notation). 

We then use the replica trick and replace (for the term of order $n$) the integral over a single matrix $A$ by integral over a $n$-uple 
of $N\times N$ Hermitian matrices $A=(A_{i})_{1\leq i\leq n}$. The replicated Gau\ss ian integral is performed with a normalized 
Gau\ss ian measure $d\mu_{C}(A)$ with a degenerated covariance  $C_{ij}=1$. Recall that for any real positive symmetric matrix 
$C_{ij}$ the Gau\ss ian integral is: 
\begin{equation}
\int d\mu_{C}(A) \,A_{i|ab}A_{j|cd}=C_{ij}\,\delta_{ad}\delta_{bc} \;,
\end{equation}
where $A_{i|ab}$ denotes the matrix element  in the row $a$ and column $b$ of the matrix $A_{i}$.
The Gau\ss ian integral with a degenerated covariance, is equivalent to inserting $n-1$ 
Dirac distributions $\delta(A_{1}\!-\!A_{2})\cdots\delta(A_{n-1}\!-\!A_{n})$, since all the $n-1$
matrices $A_{1}\!-\!A_{2}$, $\dots$,$A_{n-1}\!-\!A_{n}$ span the kernel of $C_{ij}$. This can easily be seen
by regularizing the covariance as $C_{ij}+\epsilon \delta_{ij}$ and letting $\epsilon\rightarrow 0$. Equivalently,
at the perturbative level, the uniform covariance generates all the edges (with the appropriated weights) in the Feynman graph expansion that
connect the various replicas together.

The generating function then reads:
\begin{multline}
{\cal Z}[J,J^{\dagger}]=\\\sum_{n=0}^{\infty}\frac{(-1)^{n}}{n!}
\int d\mu_{C}(A)
\prod_{i=1}^{n}\Bigg[N\Tr\log\bigg(1-\mathrm{i}\sqrt{\frac{\lambda}{N}}A_{i}\bigg)+N\Tr J\bigg(1-\mathrm{i}\sqrt{\frac{\lambda}{N}}A_{i}
\bigg)^{-1}J^{\dagger}\Bigg] \; .
\end{multline}
Remark that the Gau\ss ian measure can alternatively be written as the differential operator:
 \[
 \int d\mu_{C }(A)  \; F(A) =
 \left[ e^{\frac{1}{2} \sum_{i,j} \Tr \left[\frac{\partial }{ \partial A_i} \frac{\partial }{ \partial A_j}\right] } F(A)\right]_{A_i=0} \; .
 \]

We now apply the Bridges-Kennedy-Abdessalam-Rivasseau forest formula (see appendix \ref{BKARsec}). We start by replacing 
the covariance $C_{ij}=1$ by $C_{ij}(x)=x_{ij}$ (and $x_{ij}=x_{ji}$) evaluated at $x_{ij}=1$ for $i\neq j$ and $C_{ii}(x)=1$. Then ${\cal Z}[J,J^{\dagger}]$ is
given as a sum over forests:
\begin{multline}
{\cal Z}[J,J^{\dagger}]=
\sum_{F\,\text{labeled forest}}\frac{(-1)^{n}}{n!}\int_{0}^{1} \prod_{(i,j)\in F}dt_{ij}  \;\;   \left( \prod_{(i,j)\in F}\frac{\partial }{\partial x_{ij}} \right) \\
\times \bigg\{
\int d\mu_{C(x)}(A)\prod_{i=1}^{n}\Big[N\Tr\log\big(1-\mathrm{i}\sqrt{\frac{\lambda}{N}}A_{i}\big)+N\Tr J\bigg(1-\mathrm{i}\sqrt{\frac{\lambda}{N}}A_{i}
\bigg)^{-1}J^{\dagger}\Big]\bigg\}\bigg|_{x_{ij}=v^F_{ij}} \; ,
\label{LVEcombinatorial}
\end{multline}
 where $n$ is the number of vertices of $F$, $i$ and $j$ label the vertices of the forest, there is one weakening parameter
 $t_{ij}$ per edge $(i,j)$ of the forest and 
 \begin{equation}
v^{F}_{ij}=\left\{\begin{array}{ccl}
\inf_{(k,l)\in{ P}_{i\leftrightarrow j}^{{F}}} t_{kl}&\text{if}&  { P}_{i\leftrightarrow j}^{{F}} \; \text{exists} \\
0&\text{if}&  { P}_{i\leftrightarrow j}^{{F}} \; \text{does not exist} \
\end{array}\right. \; ,
\end{equation}
 where ${ P}_{i\leftrightarrow j}^{{F}} $ is the unique path in $F$ joining $i$ and $j$ (and the infimum is set to $1$ if $i=j$). 
 The Gau\ss ian measure 
 can alternatively be written as the differential operator:
 \[
 \int d\mu_{C(x)}(A)  F(A) =
 \left[ e^{\frac{1}{2} \sum_{i,j} x_{ij} \Tr \left[\frac{\partial }{ \partial A_i} \frac{\partial }{ \partial A_j}\right] } F(A)\right]_{A_i=0} \; .
 \]
 
 In order to extract the logarithm we use the following lemma.
\begin{lemma}
Let ${\cal W}(T)$ be a weight associated to a tree $T$, independent of the labels of its vertices and define the weight of a forest 
${\cal W} (F)$ to be the product of the weights of its trees (connected components). Then, as formal series:
\begin{equation}
\log \sum_{F \text{ labeled forests}}\frac{{\cal W}(F)}{|V(F)|!}=
\sum_{T\text{ labeled trees}}\frac{{\cal W}(T)}{|V(T)|!} \; ,
\end{equation}
with $|V(F)|$ and $|V(T)|$ the number of vertices in $F$ and $T$.
\end{lemma}
\begin{proof}
This identity is equivalent to
\begin{equation}
 \sum_{F \text{ labeled forests}}\frac{{\cal W}(F)}{|V(F)|!}=
\exp \sum_{T\text{ labeled trees}}\frac{{\cal W}(T)}{|V(T)|!} \; ,
\end{equation}
which follows by expanding the right hand side using the multinomial formula, and taking due care of the relabeling of the vertices.
\end{proof}
As both the differential operator and the Gau\ss ian measure factor over the trees in the forest $F$ we obtain:
 \begin{align*}
  \log{\cal Z}[J,J^{\dagger}]& = \sum_{T\,\text{labeled trees}}\frac{(-1)^{n}}{n!}\int_{0}^{1} \prod_{(i,j)\in T}dt_{ij} \; 
\left( \prod_{(i,j)\in T}\frac{\partial }{\partial x_{ij}} \right) 
\crcr
& \times \bigg\{
\int d\mu_{C(x)}(A)\prod_{i=1}^{n}\Big[N\Tr\log\big(1-\mathrm{i}\sqrt{\frac{\lambda}{N}}A_{i}\big)+N\Tr J\bigg(1-\mathrm{i}\sqrt{\frac{\lambda}{N}}A_{i}
\bigg)^{-1}J^{\dagger}\Big]\bigg\}\bigg|_{v^T_{ij}} \; ,\crcr
 v^{T}_{ij} & = \inf_{(k,l)\in{ P}_{i\leftrightarrow j}^{T} } t_{kl} \; .
\end{align*}
where ${ P}_{i\leftrightarrow j}^{T} $ is the unique path in the tree $T$ joining $i$ and $j$. 
Starting from the expression of the Gau\ss ian integral as a differential operator it is immediate to see that:
\begin{equation}
\frac{\partial}{\partial x_{ij}}
\bigg(\int d\mu_{C(x)}(A)F(A)\bigg)=
\int d\mu_{C(x)}(A) \; \Tr \left[\frac{\partial }{ \partial A_i} \frac{\partial }{ \partial A_j}\right] F(A) \; .
\end{equation}
This differential operator acts on two vertices ($i$ and $j$) and generates an edge connecting them. 
Taking into account that:
\begin{equation}
\frac{\partial }{ \partial A_{i|ab}}  \bigg(1-\text{i}\sqrt{\frac{\lambda}{N}}A_{i}\Big)_{cd}^{-1}
= \left( \frac{-\lambda}{N} \right) \bigg(1-\text{i}\sqrt{\frac{\lambda}{N}}A_{i}\Big)_{ca}^{-1} \bigg(1-\text{i}\sqrt{\frac{\lambda}{N}}A_{i}\Big)_{bd}^{-1} \; ,
\end{equation}
we observe that a resolvent operator $\bigg(1-\text{i}\sqrt{\frac{\lambda}{N}}A_{i}\Big)^{-1} $ is associated to each corner of a vertex.
Multiple derivatives acting on the same vertex (corresponding to multiple edges hooked to it) can act on either of the corners of the vertex
and split it into two.
The logarithm of ${\cal Z}$ becomes thus a sum over plane trees and the source terms $JJ^{\dagger}$ correspond to cilia decorating some 
of the vertices of the trees. We thus obtain:
\begin{align}
 \log {\cal Z}[J,J^{\dagger}]  = & \sum_{T\,\atop\text{LVE tree}}  {\cal A}_{T}[J,J^{\dagger},\lambda,N]  \; ,\crcr
{\cal A}_{T}[J,J^{\dagger},\lambda,N]   = &  \frac{(-\lambda)^{|E(T)|}  N^{|V(T)| - |E(T) |} }{|V(T)|!} \int_0^1 \prod_{e\in E(T)}dt_{e}\, \crcr
  & \times \int d\mu_{C_{T}}(A) \; 
 \Tr\Big[
\mathop{\overrightarrow{\prod}}\limits_{c\in \partial T\,\text{corner}}
\Big(1-\text{i}\sqrt{\frac{\lambda}{N}}A_{i_{c}}\Big)^{-1}
(JJ^{\dagger})^{\eta_{c}}\Big] \; ,
\label{LVEexpansion}
\end{align}
where $i_{c}$ is the label of the vertex the corner $c$ is attached to, $\eta_{c}\in\{0,1\}$ depending 
on whether the corner $c$ is followed ($\eta_{c}=1$) or not ($\eta_{c}=0$) by a source insertion and
the covariance $C_T$ is 
\begin{equation}
(C_{T})_{ij}=\inf_{(k,l)\in P^{T}_{i\leftrightarrow j}} t_{kl} 
\end{equation}
and the infimum is set to $1$ if $i=j$.

We have thus established the expansion \eqref{treeexpansion:eq} in Theorem \ref{treeexpansion}.
In order to establish Theorem \ref{treeexpansion} it remains to study the domain on which the expansion \eqref{LVEexpansion} is convergent. 

We first bound the amplitude of each tree using lemma \ref{boundresolvent}:
\begin{align}
\bigg|
\Tr\Big[
\mathop{\overrightarrow{\prod}}\limits_{c\in \partial T\,\text{corner}}
\Big(1-\text{i}\sqrt{\frac{\lambda}{N}}A_{i_{c}}\Big)^{-1}
(JJ^{\dagger})^{\eta_{c}}\Big]
\bigg|
&\leq
N
\mathop{\overrightarrow{\prod}}\limits_{c\in \partial T\,\text{corner}}
\Big\|
\Big(1-\text{i}\sqrt{\frac{\lambda}{N}}A_{i_{c}}\Big)^{-1}\Big\|
\Big\|JJ^{\dagger}\Big\|^{\eta_{c}}\crcr
&\leq\frac{N\|JJ^{\dagger}\|^{k}}{
\big(\cos\frac{\arg\lambda}{2}\big)^{2(|E(T)|+k)}} \; ,
\label{treeamplitudebound}
\end{align} 
where $k$ denotes the number of cilia of the tree.

Then we bound the number of LVE trees with a given number of edges and cilia.
\begin{lemma}[Counting LVE trees]
\label{coutingtrees:lem}
The number of LVE trees with $n$ edges and  $k$ cilia
\begin{equation}
{\cal N}(n,k)=\frac{(2n+k-1)!\,(n+1)!}{(n+k)!\,(n+1-k)!\,k!}
\leq 2^{2n+k-1}\,(n-1)!\,\frac{(n+1)!}{(n+1-k)!\,k!} \; ,
\end{equation}
\end{lemma}
\begin{proof}
The number of LVE trees with $n+1$ 
vertices and $k$ cilia on a fixed set of vertices labeled $i_{1},\dots,i_{k}$ is $ \frac{(2n+k-1)!}{(n+k)!}$ (see \cite{Razvannonperturbative}).
To obtain ${\cal N}(n,k)$ one simply multiplies the latter by the possible choices of $k$ vertices among $n+1$. The bound follows 
by the binomial formula $\frac{(2n+k-1)!}{(n+k)!\,(n-1)!}\leq 2^{2n+k-1}$.
\end{proof}

Consequently, the sum over LVE trees is bounded by:
\begin{align}\label{eq:boundsumtrees}
\bigg|\sum_{\text{LVE tree }} {\cal A}_{T}[J,J^{\dagger},\lambda,N]\bigg|
&\leq 
\sum_{n=0}^{\infty}\sum_{k=1}^{n+1}
\frac{N^{2}|\lambda|^{n}\|JJ^{\dagger}\|^{k}}{(n+1)!\big(\cos\frac{\arg\lambda}{2}\big)^{2n+k}} \; 
2^{2n+k-1}\,(n-1)!\,\frac{(n+1)!}{(n+1-k)!\,k!}\\\nonumber
&\leq N^{2}\sum_{n=0}^{\infty}
\frac{2^{2n-1}|\lambda|^{n}}{\big(\cos\frac{\arg\lambda}{2}\big)^{2n}}\bigg(1+\frac{2\|JJ^{\dagger}\|}{\cos\frac{\arg\lambda}{2}}\bigg)^{n+1} \; .
\end{align}

Each ${\cal A}_{T}[J,J^{\dagger},\lambda,N]$ is analytic in the cut plane ${\mathbb C}-{\mathbb R}^{-}$. Furthermore,
For every $\lambda$ inside the cardioid:
\[
 {\cal C} = \left\{  \lambda\in \mathbb{C} \;, \qquad  4|\lambda|< \cos^{2}\Big({\frac{\arg\lambda}{2}}\Big)  \right\} \; ,
\]
it is possible to find a $\epsilon_{\lambda}>0$ such that,
\begin{equation}
\frac{4|\lambda|}{\big(\cos\frac{\arg\lambda}{2}\big)^{2}}\bigg(1+\frac{2\epsilon_{\lambda} }{\cos\frac{\arg\lambda}{2}}\bigg) <1 \; ,
\end{equation}
hence theorem \ref{treeexpansion} follows.

Remark furthermore that for $\lambda \in {\cal C}$ and $\|JJ^{\dagger}\|<\epsilon_{\lambda}$, $\log {\cal Z}[J, J^{\dagger}]$ 
is analytic in $\lambda$.

\subsection{Perturbative expansion with remainder (proof of Theorem \ref{perturbativegenerating:thm})}\label{algorithm}

Our starting point is eq. \eqref{LVEexpansion}:
\begin{align}
 \log {\cal Z}[J,J^{\dagger}]  = & \sum_{T\,\atop\text{LVE tree}}  {\cal A}_{T}[J,J^{\dagger},\lambda,N]  \; ,\crcr
{\cal A}_{T}[J,J^{\dagger},\lambda,N]   = &  \frac{(-\lambda)^{|E(T)|}  N^{|V(T)| - |E(T) |} }{|V(T)|!} \int_0^1 \prod_{e\in E(T)}dt_{e}\, \crcr
  & \times \int d\mu_{C_{T}}(A) \; 
 \Tr\Big[
\mathop{\overrightarrow{\prod}}\limits_{c\in \partial T\,\text{corner}}
\Big(1-\text{i}\sqrt{\frac{\lambda}{N}}A_{i_{c}}\Big)^{-1}
(JJ^{\dagger})^{\eta_{c}}\Big] \; ,
\end{align}
and the Gaussian measure $d\mu_{ C_{T}}(A) $ can also be written as:
\[
\int d\mu_{ C_{T}}(A) \;  F(A) = \left[  e^{\frac{1}{2} \sum_{ij} \bigl( \inf_{(k,l)\in P^T_{i\leftrightarrow j} } t_{kl} \bigr)
\Tr \; \left[ \frac{\partial}{\partial A_i} \frac{\partial}{\partial A_j} \right] } \; F(A) \right]_{ A_i=0 } \; .
\]
A Taylor expansion at first order with an uniform parameter of the Gaussian measure leads to:
\begin{align*}
&  e^{\frac{1}{2} \sum_{ij}[ \inf_{(k,l)\in P^T_{i\leftrightarrow j} } t_{kl} ]
\Tr \; \left[ \frac{\partial}{\partial A_i} \frac{\partial}{\partial A_j} \right] } = 
  e^{\frac{s}{2} \sum_{ij} \bigl( \inf_{(k,l)\in P^T_{i\leftrightarrow j} } t_{kl} \bigr)
\Tr \; \left[ \frac{\partial}{\partial A_i} \frac{\partial}{\partial A_j} \right] } \Big{|}_{s=1} \crcr
& =1 + \int_0^1 ds_1 \;\; \left[ \frac{d}{ds}   e^{\frac{s}{2} \sum_{ij} \bigl(  \inf_{(k,l)\in P^T_{i\leftrightarrow j} } t_{kl} \bigr)
\Tr \; \left[ \frac{\partial}{\partial A_i} \frac{\partial}{\partial A_j} \right] }  \right]_{s=s_1} \crcr
& =1 + \int_0^1 ds_1 \;\; \left( \frac{1}{2} \sum_{ij} \bigl( \inf_{(k,l)\in P^T_{i\leftrightarrow j} } t_{kl} \bigr) 
\Tr \; \left[ \frac{\partial}{\partial A_i} \frac{\partial}{\partial A_j} \right]  \right)  
e^{\frac{s_1}{2} \sum_{ij} \bigl( \inf_{(k,l)\in P^T_{i\leftrightarrow j} } t_{kl} \bigr)
\Tr \; \left[ \frac{\partial}{\partial A_i} \frac{\partial}{\partial A_j} \right] }  
\; .
\end{align*}

The term with the Gaussian measure set to $1$ corresponds to setting all the replicated fields $A_i=0$. Consequently all the resolvents in the trace 
are replaced by the identity and the trace becomes just a trace over a product of the external sources. 

The rest term is more involved.
The new derivatives with respect to the replicated fields $A_i$ and $A_j$ act on the resolvents in the trace. As before, 
a $\frac{1}{2}\Tr \; \left[ \frac{\partial}{\partial A_i} \frac{\partial}{\partial A_j} \right]  $ operator creates a ribbon edge and brings 
an overall factor $\frac{-\lambda}{N}$.
As the edge connects two vertices already present in the tree, the new edge is necessarily a loop edge. The sums over $i$ and $j$ yields
a sum over all the possible ways to 
add such a loop edge to the tree $T$, hence we obtain a sum over all the LVE graphs $(G,T)$ one can build over $T$ having 
$|L(G,T)|=1$ loop edges. 

Iterating $L$ times we obtain:
\begin{align}\label{eq:LVE-Lloops}
& {\cal A}_{T}[J,J^{\dagger},\lambda,N]   =  \sum_{ \genfrac{}{}{0pt}{}{G } { (G,T) \text{ LVE graph} , \; |L(G,T)| \le L-1  } }
\frac{(-\lambda)^{|E(G)|}N^{\chi(G)}}{ | V(T)| !  }\prod_{f\in B(G)}\Tr\Big[\big(JJ^{\dagger}\big)^{c(f)}\Big]  \crcr
   & \qquad \times \int_0^1 \prod_{e\in E(T)}dt_{e} \left(  \prod_{e=(i,j)\in L(G,T)} \inf_{(k,l)\in P^T_{i\leftrightarrow j} } t_{kl}  \right) 
   \mathop{\int}\limits_{1\geq s_{1}\geq\cdots\geq s_{|L(G,T)|}\geq 0}\,\prod_{e\in L(G,T)}ds_{e} \crcr
& \qquad + \sum_{ \genfrac{}{}{0pt}{}{G } { (G,T) \text{ LVE graph} , \; |L(G,T)| = L  } } {\cal A}_{(G,T)}[J,J^{\dagger};\lambda,N] \; ,
\end{align}
where:
\begin{align}\label{eq:LVE-Lloops1}
& {\cal A}_{(G,T)}[J,J^{\dagger};\lambda,N]=\frac{(-\lambda)^{|E(G)|}N^{|V(G)|-|E(G)|}}{|V(G)|!}
\mathop{\int}\limits_{1\geq s_{1}\geq\cdots\geq s_{|L(G,T)|}\geq 0}\,\prod_{e\in L(G,T)}ds_{e} \crcr
& \times \mathop{\int}\limits_{[0,1]} \prod_{e\in E(T)} dt_{e}
  \left( \prod_{e=(i,j)\in L(G,T)}\mathop{\text{inf}}\limits_{e'\in P_{i\leftrightarrow j}^{T}} t_{e'}  \right) \crcr
& \times \int d\mu_{s_{|L(G,T)|}C_{T}}(A)\prod_{f\in F(G)}\Tr\bigg\{\mathop{\prod}\limits_{c\in\partial f}^{\longrightarrow}
\bigg(1-\text{i}\sqrt{\frac{\lambda}{N}}\,A_{i_{c}}\bigg)^{-1}(JJ^{\dagger})^{\eta_{c}}\bigg\} \; .
\end{align}
 yielding the expression eq. \eqref{LVEamplitude} for the amplitude of an LVE graph.

We prove theorem \ref{perturbativegenerating:thm} by induction on $n$. 
We start with the tree expansion in \eqref{treeexpansion:eq} and select the unique LVE tree  without any edge. 
This tree is a single vertex with one cilium (the vertex without a cilium is absent because of
the normalization ${\cal Z}[0,0,\lambda, N]$). We perform an expansion up to $L=1$ loop edges for this term.
All the other LVE trees with edges 
are included in the rest term. 
We obtain:
\begin{align*}
 \log {\cal Z}[J,J^{\dagger}]  = & N \Tr[JJ^{\dagger}] + 
                         \sum_{ \genfrac{}{}{0pt}{}{G } { (G,T) \text{ LVE graph} , \; | V(G) | =1, \; |L(G,T)| =1   } } {\cal A}_{(G,T)}[J,J^{\dagger};\lambda,N] \crcr
                                 & + \sum_{\genfrac{}{}{0pt}{}{T \; \text{LVE tree}}{ E(T) \ge 1 }} {\cal A}_{T}[J,J^{\dagger},\lambda,N]  \; .
\end{align*}
As the trees with exactly one edge are LVE graphs themselves we can move them to the first rest term and write:
\begin{align*}
 \log {\cal Z}[J,J^{\dagger}]  = & N \Tr[JJ^{\dagger}] + 
                         \sum_{ \genfrac{}{}{0pt}{}{G } { (G,T) \text{ LVE graph} , \; E(G)  =1   } } {\cal A}_{(G,T)}[J,J^{\dagger};\lambda,N] \crcr
                                 & + \sum_{\genfrac{}{}{0pt}{}{T \; \text{LVE tree}}{ E(T) \ge 2 }} {\cal A}_{T}[J,J^{\dagger},\lambda,N]  \; ,
\end{align*}
reproducing eq. \eqref{perturbativeexpansion:eq} for $n=0$.
The first term in this expression is the contribution of an ordinary Feynman graph (without any resolvent). If we write the theory
in terms only of $M$, this graph has two univalent vertices $J$ and $J^{\dagger}$ connected by an edge.
The second term is the amplitude for a LVE graph with one vertex, one cilium and one edge (which can either be a loop edge or a tree edge). 
The last term is the contribution of all the LVE trees with at least two edges. 

Let us assume that the theorem has been established up to order $n$ hence the perturbative remainder at order $n$,
$R_{n}[J,J^{\dagger},\lambda,N]$ is the sum of two terms:
\begin{align*}
R'_{n}[J,J^{\dagger},\lambda,N]&=
\sum_{ \genfrac{}{}{0pt}{}{ (G,T)\,\text{LVE graph}} {  |E(G)|=n+1 } }{\cal A}_{(G,T)}[J,J^{\dagger},\lambda,N] \; , \crcr
R''_{n}[J,J^{\dagger},\lambda,N]&=
\sum_{ \genfrac{}{}{0pt}{}{ T\,\text{LVE trees} }{ |E(T)|\ge n+2 } } {\cal A}_{T}[J,J^{\dagger},\lambda,N] \; .
\end{align*}

The trees contributing to $R_n''$ having at least $n+3$ edges give exactly $R''_{n+1}$ (where we omit the arguments in order to simplify the notation). 
The trees having exactly $n+2$ edges are transferred to $R_{n+1}'$. Now consider the 
LVE graphs in  $R'_n$. The amplitude of each of these graphs is written as a Gau\ss ian integral with covariance $s_{|L(G,T)|} C_T$. 
For each of these graphs we expand one more loop edge using:
\begin{align*}
& e^{\frac{s_{|L(G,T)|}}{2} \sum_{ij} \bigl( \inf_{(k,l)\in P^T_{i\leftrightarrow j} } t_{kl} \bigr)
\Tr \; \left[ \frac{\partial}{\partial A_i} \frac{\partial}{\partial A_j} \right] } \Big{|}_{s=1} \crcr
& =1 + \mathop{\int} \limits_{s_{|L(G,T)| }  \ge s_{|L(G,T)| +1} \ge 0 } ds_{|L(G,T)| +1}
 \;\; \left( \frac{1}{2} \sum_{ij} \bigl( \inf_{(k,l)\in P^T_{i\leftrightarrow j} } t_{kl} \bigr) 
\Tr \; \left[ \frac{\partial}{\partial A_i} \frac{\partial}{\partial A_j} \right]  \right)  \crcr
& \qquad \qquad \times e^{\frac{s_{|L(G,T)| +1} }{2} \sum_{ij} \bigl( \inf_{(k,l)\in P^T_{i\leftrightarrow j} } t_{kl} \bigr)
\Tr \; \left[ \frac{\partial}{\partial A_i} \frac{\partial}{\partial A_j} \right] }  
\; .
\end{align*}
The rest terms are all collected to yield the remaining terms in $R'_{n+1}$, as they are all LVE graphs of order $n+2$, with the correct amplitude. 

It remains to check that the evaluation of the new explicit terms reproduces exactly the perturbative 
evaluation of the amplitude of the graphs with exactly $E(G) = n+1$ edges.  
As all the $A_{i}$'s are set to zero in the explicit terms, the product of traces yields just:
\[
 N^{|F(G)|-|B(G)| } \prod_{f\in B(G)}\Tr\Big[\big(JJ^{\dagger}\big)^{c(f)}\Big] \;.
\]

The integral over the loop parameters $s_e$ can be trivially performed:
\begin{equation}
\mathop{\int}\limits_{1\geq s_{1}\geq\cdots\geq s_{|L(G,T)|}\geq 0}\,\prod_{e\in L(G,T)}ds_{e}=\frac{1}{|L(G,T)|!} \; .
\end{equation}
Since there are precisely $|L(G,T)|!$ ways to label the loop edges, the sum becomes a sum over graphs with unlabeled loop edges. 
We are left with ciliated ribbon graphs with labels on their vertices and a distinguished spanning tree. 

The integral over the weakening parameters $t_e$ is more subtle (see \cite{Vincentresum}). A Hepp sector $\alpha$ in a graph $G$ 
is a total order of the edges of $G$. For any Hepp sector $\alpha$, the dominant tree $T(\alpha)$ in the sector is obtained by iteratively choosing the 
highest edges in the sector which do not form cycles. It turns out (\cite{Vincentresum}) that for fixed $G$ and $T$ the integral 
over the parameters $t_e$ is the percentage of Hepp sectors of $G$ in which the tree $T$ is dominant.  The following lemma is then trivial.
\begin{lemma}
\label{Heppsectors}
For any vertex labeled graph $G$:
\begin{equation}
\sum_{\genfrac{}{}{0pt}{}{T\subset G}{ T\, \text{ spanning tree} } }
\int \prod_{e\in E(T)}dt_{e}\, \left( \prod_{e=(i,j)\in L(G,T)}
\inf_{e'\in P^T_{i\leftrightarrow j}}t_{e'} \right) =1 \; .
\end{equation}
\end{lemma}

Collecting these two results together, the explicit terms are a sum over vertex labeled ciliated ribbon graphs with exactly $n+1$ edges. 
Taking into account the amplitude does not depend on the labeling of the vertices, we collect together all the terms corresponding to 
different labellings of the same ribbon graph and use:
\begin{align*}
& \frac{1}{|V(G)|!}\sum_{\text{labellings of $V(G)$}}
(-\lambda)^{|E(G)|}N^{\chi(G)}\prod_{f\in B(G)}\Tr\Big[\big(JJ^{\dagger}\big)^{c(f)}\Big] \crcr
& =\frac{(-\lambda)^{|E(G)|}N^{\chi(G)}}{|\text{Aut}(G)|}\prod_{f\in B(G)}\Tr\Big[\big(JJ^{\dagger}\big)^{c(f)}\Big] \; ,
\end{align*}
where $|\text{Aut G}|$ corresponds to the order of the group of permutations of the vertices that leave the graph invariant
to obtain the explicit terms in eq. \eqref{perturbativeexpansion:eq}.

Finally, the analyticity of the remainder is obvious since $R'_{n}[J,J^{\dagger},\lambda,N]$ is a finite sum of analytic functions in 
$\mathbb{C}-\mathbb{R}^-$ and $R''_{n}[J,J^{\dagger},\lambda,N]$ is bounded by the bound in eq. \ref{eq:boundsumtrees}, hence 
converges and is analytic under the same hypothesis.
  
\subsection{Topological expansion (proof of Theorem \ref{topologicalexpansion:thm})} 

In the expansion theorem \ref{perturbativegenerating:thm}, we have recursively added loop edges to the trees 
irrespective of the genus of the graph $(G,T)$ we obtained. 
In order to prove the topological expansion theorem \ref{topologicalexpansion:thm}, we use the same algorithm, except that we stop 
adding loop edges to a graph if its genus reaches $g+1$. We thus obtain:
\begin{multline}
\log{\cal Z}[J,J^{\dagger};\lambda,N]=
\sum_{G\text{ ciliated ribbon graph}\atop |E(G)|\leq n \text{ and }g(G)\leq g}
\frac{(-\lambda)^{|E(G)|}N^{2-2g(G)-B(G)}}{|\text{Aut}(G)|}\prod_{f\in B(G)}\Tr\Big[\big(JJ^{\dagger}\big)^{c(f)}\Big]\\
+\widetilde{{\cal R}}_{g,n}[J,J^{\dagger};\lambda,N]
+\widetilde{{\cal R}}_{g,n}^{'}[J,J^{\dagger};\lambda,N]
+\widetilde{{\cal R}}_{n}^{''}[J,J^{\dagger};\lambda,N] \;,
\end{multline}
where now there are three classes of remainder terms.

The first remainder term is made of LVE graphs with less than $n+1$ edges such that the addition of the last loop edge (with label $|L(G,T)|$)
increases the genus form $g$ to $g+1$:
 \begin{equation}
\widetilde{\cal R}_{g,n}[J,J^{\dagger};\lambda,N]=
\sum_{(G,T)\text{ LVE graph with }|E(G)|\leq n+1\atop
g(G)=g+1\text{ and  }g(G-e_{|L(G,T)|})=g}
{\cal A}_{(G,T)}[J,J^{\dagger};\lambda,N] \;. 
\end{equation}
The second remainder term is a summation over LVE graphs with $n+1$ edges and genus less than $g$:
\begin{equation}
\widetilde{\cal R}'_{g,n}[J,J^{\dagger};\lambda,N]=
\sum_{(G,T)\text{ LVE graph with }\atop
|E(G)|= n+1\text{ and }g(G)\leq g}
{\cal A}_{(G,T)}[J,J^{\dagger};\lambda,N] \; .
\end{equation}
Finally, the last remainder term is a sum over LVE trees with at least $n+2$ edges on which the loop generating algorithm has not yet been applied: 
\begin{equation}
\widetilde{{\cal R}}_{n}''[J,J^{\dagger};\lambda,N]=
\sum_{T\text{ LVE tree}\atop |E(T)|\geq n+2}
{\cal A}_{T}[J,J^{\dagger};\lambda,N] \; .
\end{equation}

In order to take the limit $n\rightarrow\infty$, we bound the LVE amplitudes ${\cal A}_{(G,T)}[J,J^{\dagger};\lambda,N]$ as well as 
the number of graphs contributing to the remainders.

To bound the LVE amplitude in eq. \eqref{LVEamplitude} we observe that the latter is a product over faces of traces of products of resolvents. 
Bounding each trace as $|\Tr(O)|\leq N\|O\|$ and using lemma \ref{boundresolvent} for the norm of the resolvents, we get
\begin{multline}
\bigg|{\cal A}_{(G,T)}[J,J^{\dagger},\lambda,N]\bigg|\leq \int \prod_{e\in E(T)} dt_{e}
  \left( \prod_{e=(i,j)\in L(G,T)}\mathop{\text{inf}}\limits_{e'\in P_{i\leftrightarrow j}^{T}} t_{e'} \right) \\
\frac{N^{|F(G)|+|V(G)|-|E(G)|}|\lambda|^{|E(G)|}}{|V(G)|! |L(G,T)| !}\bigg(\frac{1}{\cos\frac{\arg\lambda}{2}}\bigg)^{2E(G)+k}\|JJ^{\dagger}\|^{k} \; .
\end{multline}
This bound is very similar to the one of the tree amplitude \eqref{treeamplitudebound}, except that we get one factor of $N$ for each internal face $G$   
and the integral over loop parameters $s_e$ yields a factor $\frac{1}{ |L(G,T)| !}$. 
Note that we have left the integral over the weakening parameters $t_{e}$ since it allows to cancel the choice of the spanning tree. Let us denote by
$\tilde{\cal N}(g, n,k)$ the number of ribbon graphs with unlabeled vertices having  
genus $g$, $n$ edges and $k$ cilia. Using proposition \ref{Heppsectors} and noticing that $\frac{1}{|L(G,T)|!}$ cancels the labeling of the loop edges, we get:
\begin{multline}
\sum_{(G,T)\text{ LVE graph with $k$ cilia }\atop
|E(G)|= n\text{ and }g(G)=g}
\int \prod_{e\in E(T)} dt_{e}
 \left( \prod_{e=(i,j) \in L(G,T)}\mathop{\text{inf}}\limits_{e'\in P_{i\leftrightarrow j}^{T}} t_{e'} \right) \quad
\frac{1}{|V(G)|! |L(G,T)| !}\\=
\sum_{G\text{ ribbon graph with labeled vertices}
\atop \text{$|E(G)|=n$  and $g(G)=g$}}\frac{1}{|V(G)|!}
\leq \widetilde{{\cal N}}(g,n,k) \; .
\end{multline}
The inequality (instead of an equality) in the last line comes from the possibility to have different labellings of the vertices leading to the same unlabeled graph 
(this is also the origin of the factor $1/|\text{Aut}(G)|  $ in the explicit terms).

Now we need a bound on $ \widetilde{{\cal N}} (g,n,k)$. We obtain all graph with $k$ cilia by adding $k-1$ cilia to a graph with $1$ cilia.
Relaxing the condition that there is at most one cilium per vertex to the condition that two cilia are not adjacent, in a graph with $n>0$ edges and 1
cilium, there are $2n+1-2=2n-1$ corners on which we can add the second cilium (the graph has $2n+1$ corners, but the 2 corners adjacent to the first cilium 
are forbidden). When adding the new cilium we create a new corner, but the two corners adjacent to the new cilium 
(which are distinct since $n>0$) are forbidden. 
Therefore, there are $2n-2$ corners on which we can add a third cilium and so on up to the last cilium for which we have $2n-(k-1)=2n+1-k$ available corners
Therefore,
\begin{equation}
\widetilde{\cal N}(g,n,k)\leq
\frac{\overbrace{(2n-1)\dots(2n+1-k)}^{k-1\text{ terms}}}{k!} \; \widetilde{\cal N}(g,n,1)\leq
\frac{(2n)!}{k!(2n-k)!} \; \widetilde{\cal N}(g,n,1) \;
\end{equation}
where we have divided by $k!$ since all the cilia are indistinguishable (i.e. adding cilia on the same corners but in a different order leads to the same graph). 
Note that we only obtain an inequality since adding cilia to different graphs can lead to the same ciliated graph.

The number of genus $g$ graphs with $n$ edges and a single cilium is equal to the number of rooted bipartite quadrangulations (due to the bijection 
in proposition \ref{dualbijection} applied to maps with one marked edge) which is known \cite{SchaefferChapuyMarcus}. 

\begin{lemma} \label{schaeffertheorem}
The number $\widetilde{{\cal N}}(g,n)$ of rooted maps with $n$ edges and genus $g$ has the asymptotic behavior
\begin{equation}
\widetilde{{\cal N}}(g,n)\mathop{\sim}\limits_{n\rightarrow\infty}
C_{g}12^{n} n^{\frac{5}{2}(g-1)} \;, 
\end{equation} 
with $C_{g}$ a constant that only depends on the genus.
\end{lemma}

Consequently, there is a constant $C'_{g}$ such that, for $n$ large enough $\widetilde{{\cal N}}(g,n)=\widetilde{{\cal N}}(g,n,1)\leq
C'_{g}12^{n} n^{\frac{5}{2}(g-1)}$, so that
\begin{equation}
\widetilde{\cal N}(g,n,k)\leq
C'_{g}12^{n} n^{\frac{5}{2}(g-1)}
\frac{(2n)!}{k!(2n-k)!} \; .
\label{ciliabound:eq}
\end{equation}
We thus obtain the bound:
\begin{align}
&\sum_{(G,T)\text{ LVE graph with }\atop
|E(G)|= n\text{ and }g(G)=g}
\big|{\cal A}_{(G,T)}[J,J^{\dagger},\lambda,N]\big|
\nonumber\\&\quad \leq
 \sum_{k=1}^{n+1}
N^{2-2g}|\lambda|^{n}
\bigg(\frac{1}{\cos\frac{\arg\lambda}{2}}\bigg)^{2n+k}\|JJ^{\dagger}\|^{k} \;  \frac{(2n)!}{k!(2n-k)!} \;
\widetilde{{\cal N}}(g,n)\nonumber\\
&\quad\leq C'_{g}12^{n} n^{\frac{5}{2}(g-1)}
N^{2-2g}\bigg(\frac{|\lambda|}{\cos^{2}\frac{\arg\lambda}{2}}\bigg)^{n}
\bigg(1+\frac{\|JJ^{\dagger}\|}{\cos\frac{\arg\lambda}{2}}\bigg)^{2n} \; ,                                                                                      
\end{align}
where we have extended the sum over $k$ from 0 to $2n$ (instead of $n+1$) and used the binomial formula.

For every $\lambda \in \widetilde{\cal C}$, there exists $\epsilon_{\lambda}>0$ 
such that 
\begin{equation}
\bigg(\frac{|\lambda|}{\cos^{2}\frac{\arg\lambda}{2}}\bigg)
\bigg(1+\frac{  \epsilon_{\lambda} }{\cos\frac{\arg\lambda}{2}}\bigg)^2=\xi<\frac{1}{12} \; .
\end{equation}
We chose $ \| JJ^{\dagger} \| \le \epsilon_{\lambda}$. We then have the following bounds:
\begin{itemize}
 \item we bound the term 
 \[ 
\widetilde{\cal R}_{g,n}[J,J^{\dagger};\lambda,N]=
\sum_{(G,T)\text{ LVE graph with }|E(G)|\leq n+1\atop
g(G)=g+1\text{ and  }g(G-e_{|L(G,T)|})=g}
{\cal A}_{(G,T)}[J,J^{\dagger};\lambda,N] \;. 
 \]
by a sum over graphs of genus $g+1$ having at most $n+1$ edges:
\begin{align*}
 | \widetilde{\cal R}_{g,n}[J,J^{\dagger};\lambda,N] | \le C'_{g+1} N^{2-2(g+1)} \sum_{m=0}^{n+1} m^{\frac{5}{2} g } (12 \, \xi)^m \;,
 \end{align*}
which is convergent as $n\to \infty$ while keeping $N$ fixed.
 \item we bound the term 
\[
\widetilde{\cal R}'_{g,n}[J,J^{\dagger};\lambda,N]=
\sum_{(G,T)\text{ LVE graph with }\atop
|E(G)|= n+1\text{ and }g(G)\leq g}
{\cal A}_{(G,T)}[J,J^{\dagger};\lambda,N] \; .
\]
by:
\begin{align*}
 | \widetilde{\cal R}'_{g,n}[J,J^{\dagger};\lambda,N] | \le N^2 \left(  \sum_{h=0}^g \frac{C'_h}{N^{2h}}  (n+1)^{\frac{5}{2}(h-1)} \right) (12 \, \xi)^{n+1} \; ,
 \end{align*}
hence this term goes to zero when sending $n \to \infty$ while keeping $N$ fixed.
 \item as the sum over LVE trees is convergent, the last reminder term:
\[
\widetilde{{\cal R}}_{n}''[J,J^{\dagger};\lambda,N]=
\sum_{T\text{ LVE tree}\atop |E(T)|\geq n+2}
{\cal A}_{T}[J,J^{\dagger};\lambda,N] \; ,
\]
also goes to zero when sending $n \to \infty$ while keeping $N$ fixed.
\end{itemize}
This achieves the proof of theorem  \ref{topologicalexpansion:thm}.

\section{Proofs of the theorems regarding the cumulants}\label{sec:proof2}

\subsection{Cumulants and their structure (proofs of Propositions \ref{prop:ampliWeingarten} and \ref{structure:prop})}
 
Before establishing the proposition \ref{prop:ampliWeingarten} we detail some properties of the Weingarten functions.

\begin{lemma}[Convolution inverse]
For $N>k$ and any permutations $\sigma,\tau\in{\mathfrak S}_{k}$, one has:
\begin{equation}
\sum_{\tau\in{\mathfrak S}_{k}}  N^{|C(\tau \rho^{-1})|}  \; \; \text{Wg}(\tau\sigma^{-1},N) 
=\begin{cases}
1&\text{if $\rho=\sigma$}\\
0&\text{otherwise}
\end{cases}
\end{equation}
where $|C(\sigma)|$ is the number of cycles in the decomposition of $\sigma$.
\end{lemma}
\begin{proof}
For any permutation $\rho$ we have:
 \begin{align}
& \sum_{a,c} \left( \prod_{i=1}^k \delta_{a_{\rho(i)} c_i } \right)  \int dU \; U_{a_{1}b_{1}}\dots U_{a_{k}b_{k}}
U^{*}_{c_{1}d_{1}}\dots U^{*}_{c_{k}d_{k}} \crcr
& = \sum_{a,c} \left( \prod_{i=1}^k \delta_{a_{\rho(i)} c_i } \right) 
\sum_{\sigma,\tau\in \mathfrak{S}_{k}} \left( \prod_{i=1}^k  \delta_{a_{\tau(i)} c_i} \delta_{ b_{\sigma(i)} d_i }\right) \text{Wg}(\tau\sigma^{-1},N) 
\Rightarrow \crcr
&  \left( \prod_{i=1}^k \delta_{b_{\rho(i)} d_i } \right) = \sum_{\sigma,\tau\in \mathfrak{S}_{k}} N^{|C(\tau \rho^{-1})|}
\left( \prod_{i=1}^k  \delta_{ b_{\sigma(i)} d_i }\right) 
\text{Wg}(\tau\sigma^{-1},N) \; .
 \end{align}
Applying this equality for $b_{\sigma(i)} = d_i=i$ the left hand side is non zero only for $\rho=\sigma$ and in this case it equals $1$.
\end{proof}

In our context, the Weingarten functions are used in order to write any unitary invariant
homogeneous polynomial of degree $k$ of a $N\times N$ hermitian matrix $H$ as a linear combination of products of
traces of powers of $H$. Let $P$ be such a polynomial:
\begin{equation}
P(H)=\sum_{1\leq p_{1},q_{1},\dots,p_{k},q_{k}\leq N}
A_{p_{1},q_{1},\dots,p_{k},q_{k}}\,
H_{p_{1}q_{1}}\cdots H_{p_{k}q_{k}} \;,
\end{equation}
with $P(UHU^{\dagger})=P(H)$ for any $U\in\text{U}(N)$.

\begin{lemma}[Expansion over trace invariants]
\label{traceinvariants:prop}
Any unitary invariant degree $k$ homogeneous polynomial can be written as:
\begin{equation}
P(H)=\sum_{\pi\in\Pi_{k}}P_{\pi}\,\Tr_{\pi}(H) \;, 
\end{equation}
with:
\begin{equation}
P_{\pi}=\sum_{\sigma,\tau\in{\mathfrak S}_{k}\atop
C(\sigma)=\pi}
\sum_{1\leq p_{1},\dots,p_{k}\leq N}
A_{p_{1}p_{\tau(1)},\dots,p_{k}p_{\tau(k)}}
\text{Wg}(\tau\sigma^{-1},N)
\end{equation}
and the trace invariant (that only depend on the cycle structure $C(\sigma)$) are:
\begin{equation}
\Tr_{\pi}(H)=\sum_{1\leq a_{1},\dots,a_{k}\leq N}
H_{a_{1}a_{\sigma(1)}}\dots H_{a_{k}a_{\sigma(k)}} \;. 
\end{equation}

In particular, if $P$ is already a trace invariant associated to the partition $\pi_{0}$ then 
$P_{\pi}=1$ if $\pi=\pi_{0}$ and $0$ otherwise.
\end{lemma}

\begin{proof}
 Due to unitary invariance of $P$ we have 
 \begin{align*}
&  P(H) = \int [dU] \; P(UHU^{\dagger}) = \sum_{1\leq p_{1},q_{1},\dots \leq N} 
A_{p_{1},q_{1},\dots,p_{k},q_{k}}\, \left(  \prod_{i=1}^k U_{p_i a_i} U^*_{q_i b_i} H_{a_{i}b_{i}}  \right) \crcr
&= \sum_{\tau, \sigma \in \mathfrak{S}_k } \left(  
 \sum_{1\leq p_{1},q_{1},\dots \leq N}  A_{p_{1},q_{1},\dots,p_{k},q_{k}}  \prod_{i=1}^k  \delta_{p_{\tau(i)} q_i } \right)  
 \left(  \prod_{i=1}^k \delta_{a_{\sigma(i)} b_i } H_{a_{i}b_{i}}  \right) \text{Wg}(\tau\sigma^{-1},N) \; .
 \end{align*}
If $P$ is the trace invariant associated to $\pi_0=(k_1,\dots k_{|\pi_0|})$ then:
\[
 A_{p_{1},q_{1},\dots,p_{k},q_{k}} =  \frac{1}{ \sum_{ \genfrac{}{}{0pt}{}{ \rho\in \mathfrak{S}_k}{  C(\rho)=\pi_0 } } 1 } 
  \sum_{ \genfrac{}{}{0pt}{}{ \rho\in \mathfrak{S}_k}{  C(\rho)=\pi_0 } }
 \prod_{i=1}^k \delta_{q_ip_{\rho(i)}}
\]
and 
\begin{align*}
 P_{\pi}&=\sum_{\sigma,\tau\in{\mathfrak S}_{k}\atop
C(\sigma)=\pi}
\left( \sum_{1\leq p_{1},\dots,p_{k}\leq N}
 \frac{1}{ \sum_{ \genfrac{}{}{0pt}{}{ \rho\in \mathfrak{S}_k}{  C(\rho)=\pi_0 } } 1 } 
  \sum_{ \genfrac{}{}{0pt}{}{ \rho\in \mathfrak{S}_k}{  C(\rho)=\pi_0 } }
 \prod_{i=1}^k \delta_{p_{\tau(i)} p_{\rho(i)}} \right)
\text{Wg}(\tau\sigma^{-1},N) \crcr
& =  
 \frac{1}{ \sum_{ \genfrac{}{}{0pt}{}{ \rho\in \mathfrak{S}_k}{  C(\rho)=\pi_0 } } 1 } 
  \sum_{ \genfrac{}{}{0pt}{}{ \rho\in \mathfrak{S}_k}{  C(\rho)=\pi_0 } } 
\sum_{\sigma,\tau\in{\mathfrak S}_{k}\atop 
C(\sigma)=\pi}  N^{|C(\tau \rho^{-1})|}  \; \; \text{Wg}(\tau\sigma^{-1},N) \; .
\end{align*}
By the convolution inverse identity, the sum over $\tau$ enforces $\rho = \sigma$ and the lemma follows.
 \end{proof}
We chose a permutation $\zeta\in{\mathfrak S}_{k}$ whose 
cycle decomposition reproduces the contribution of the broken faces to the amplitude of a LVE graph:
\begin{equation}
\zeta=(i_{1}^{1}\dots i_{k_{1}}^{1}) \cdots (i_{1}^{b}\dots i_{k_{b}}^{b}) \; .
\end{equation}
if there are $b=|B(G)|$ broken faces with $k_{1},\dots,k_{b}$ cilia. Denoting $X^{l}$ the product of the resolvents in between 
the cilia $l$ and $\zeta(l)$ and $Y^{m}$ the product of the resolvents around the unbroken face labeled $m$ the 
amplitude can be written as:
\begin{multline}
A_{(G,T)}[J,J^{\dagger},\lambda,N]=
\frac{(-\lambda)^{|E(G)|}N^{|V(G)|-|E(G)|}}{|V(G)|!}
\mathop{\int}\limits_{1\geq s_{1}\geq\cdots\geq s_{|L(G,T)|}\geq 0}\,\prod_{e\in L(G,T)}ds_{e}\\
\int \prod_{e\in E(T)} dt_{e}
 \left(  \prod_{e=(i,j)\in L(G,T)}\mathop{\text{inf}}\limits_{e'\in P_{i\leftrightarrow j}^{T}} t_{e'} \right) \\
\int d\mu_{s_{|L(G,T)|}C_{T}}(A)\prod_{1\leq m\leq B(G)}\Tr\Big[JJ^{\dagger}
\mathop{\prod}\limits_{1\leq r\leq k_{m}}^{\longrightarrow}X^{i^{m}_{r}}\Big]
\prod_{1\leq m\leq F(G)-B(G)}\Tr\Big[Y^{m}\Big] \;.
\end{multline}
This is a degree $k$ homogeneous polynomial which is invariant under the unitary transformation $J\rightarrow UJ$ and
$J^{\dagger}\rightarrow J^{\dagger}U^{\dagger}$. Indeed, because of the invariance of $d\mu_{C_{T}}(A)$ the transformation of
$J$ can be compensated by a transformation on the matrices $A_{1},...,A_{n}$. 
Therefore, we may apply lemma \ref{traceinvariants:prop} to expand it over trace invariants:
\begin{equation}
{\cal A}_{(G,T)}[J,J^{\dagger},\lambda,N]=\sum_{\pi\in\Pi_{k}}{A}_{(G,T)}^{\pi}(\lambda,N)\Tr_{\pi}(JJ^{\dagger}) \; ,
\end{equation}
with:
\begin{multline}
{A}_{(G,T)}^{\pi}(\lambda,N)=\frac{(-\lambda)^{|E(G)|}N^{|V(G)|-|E(G)|}}{|V(G)|!}
\mathop{\int}\limits_{1\geq s_{1}\geq\cdots\geq s_{|L(G,T)|}\geq 0}\,\prod_{e\in L(G,T)}ds_{e}\\
\int \prod_{e\in E(T)} dt_{e}
 \left( \prod_{e=(i,j) \in L(G,T)}\mathop{\text{inf}}\limits_{e'\in P_{i\leftrightarrow j}^{T}} t_{e'} \right)  
\int d\mu_{s_{|L(G,T)|}C_{T}}(A) \\ 
\times \sum_{\tau,\sigma\in{\mathfrak S}_{k}\atop C(\sigma)=\pi}
\sum_{1\leq p_{1},\dots,p_{k}\leq N}\text{Wg}(\tau\sigma^{-1},N)
\prod_{1\leq m\leq F(G)-B(G)}\Tr\Big[Y^{m}\Big]\prod_{1\leq l\leq k}
X^{l}_{p_{\tau(l) }p_{\zeta(l)}} \; .
\end{multline}

This proves proposition \ref{prop:ampliWeingarten}.

In order to prove proposition \ref{structure:prop}, we use the expression of $\log{\cal Z}[J,J^{\dagger},\lambda,N]$ 
form theorem \ref{treeexpansion}.
As the sum over LVE trees is convergent we can derive term by term with respect to the sources. 
Before deriving we express the amplitude of a LVE tree as in proposition \ref{prop:ampliWeingarten}. 
Proposition \ref{structure:prop} follows from the remark that for any two 
permutations $\tau,\xi\in{\mathfrak{S}}_{k}$ such that $\tau\xi^{-1}$ has cycle decomposition corresponding to the partition $\pi$,
we have:
\begin{equation}
\Tr_{\pi}\big(JJ^{\dagger}\big)=
\sum_{1\leq p_{1},q_1 \dots \leq N}
\prod_{1\leq l\leq k} J_{p_{l}q_{l}}J^{\ast}_{p_{\tau\xi^{-1}(l)}q_{l}}=
\sum_{1\leq p_{1},q_1 \dots \leq N}
\prod_{1\leq l\leq k} J_{p_{l}q_{l}}J^{\ast}_{p_{\tau(l)}q_{\xi(l)}} \;,
\end{equation}
hence the derivative with respect to the sources is:
\begin{multline}
\frac{\partial^{k}}{\partial J_{a_{1}b_{1}}\cdots\partial J_{a_{k}b_{k}}}
\frac{\partial^{k}}{\partial J^{\ast}_{c_{1}d_{1}}\cdots\partial J^{\ast}_{c_{k}d_{k}}}\Tr_{\pi}\big(JJ^{\dagger}\big)= \\
=\sum_{\rho,\sigma\in\mathfrak{S}_{k}}
\sum_{1\leq p_{1},q_1 \dots\leq N}
\prod_{1\leq l\leq k} 
\delta_{a_{\rho(l)},p_{l}}\delta_{b_{\rho(l)},q_{l}}
\delta_{c_{\sigma(l)},p_{\tau(l)}}\delta_{d_{\sigma(l)},q_{\xi(l)}} \; ,
\end{multline}
and summing over $p_{l}$ and $q_{l}$ we obtain:
\begin{equation}
\frac{\partial^{k}}{\partial J_{a_{1}b_{1}}\cdots\partial J_{a_{k}b_{k}}}
\frac{\partial^{k}}{\partial J^{\ast}_{c_{1}d_{1}}\cdots\partial J^{\ast}_{c_{k}d_{k}}}\Tr_{\pi}\big(JJ^{\dagger}\big)=
\sum_{\rho,\sigma\in\mathfrak{S}_{k}}
\prod_{1\leq l\leq k}\delta_{c_{l},a_{\rho\tau\sigma^{-1}(l)}}\delta_{d_{l},b_{\rho\xi\sigma^{-1}(l)}} \; .
\end{equation}

\subsection{Constructive theorems for cumulants (proofs of Theorems \ref{treecumulants:thm} and \ref{perturbativecumulants:thm})}
 
In order to prove the constructive theorem for the cumulants we need to bound the amplitude in eq. \eqref{amplitudeWg}. 
The contribution of the unbroken faces and of the broken faces are made of products of resolvents. The summation 
over the indices reproduces a product of $|C(\zeta\tau^{-1})|$ traces for the broken faces and $F(G)-B(G)$ traces for the unbroken faces.

We bound the Weingarten functions using the following lemma.
\begin{lemma}
\label{Weingartenbound:lem}
For $N$ large enough,
\begin{equation}
 | \text{Wg}(\sigma,N) | <\frac{2^{2k}}{N^{2k-|C(\sigma)|}} \; .
\label{Weingartenbound}
\end{equation}
\end{lemma}
This lemma can be deduced from the asymptotic behavior of the Weingarten functions \cite{Collins, ColSni}.

We bound the norm of the resolvents using lemma \ref{boundresolvent} (recall that we have a resolvent per corner and there 
are $2|E(G)|+k$ corners on a LVE graph with $|E(G)|$ edges and $k$ cilia). Taking into account that each trace produces a factor of $N$ 
we obtain:
 \begin{multline}
 \bigg|\text{Wg}(\tau\sigma^{-1},N)
\sum_{1\leq p_{1},\dots,p_{k}\leq N}
\prod_{1\leq m\leq f-b}\Tr\Big[Y^{m}\Big]\prod_{1\leq l\leq k}
X^{l}_{p_{l}p_{\zeta\tau^{1}(l)}}
\bigg| \\ 
\leq \frac{2^{2k}N^{|C(\zeta\tau^{-1})|+|C(\tau\sigma^{-1})|-2k+F(G)-B(G)}}{\big(\cos\frac{\arg\lambda}{2}\big)^{2E(G)+k}} \; .
 \end{multline}

In order to bound the scaling in $N$ of the amplitude of a LVE graph we use the following lemma \cite{Razvannonperturbative}.
\begin{lemma}
\label{permutationinequalitylemma}
Let $\sigma$ and $\tau$ two permutations of $k$ elements. Then, the numbers of cycles in the decompositions of $\sigma$, $\tau$ and $\sigma\tau$  obey:
\begin{equation}
|C(\sigma)|+|C(\tau)|\leq k+|C(\sigma\tau)| \; .
\label{permutationinequality}
\end{equation}
\end{lemma}
\begin{proof} We will prove the following more general inequality: for any three permutations $\sigma, \tau$ and $\xi$ of $k$ elements we have:
\[
 |C(\sigma \xi^{-1} )| + |C(\xi \tau)| \le k + |C(\sigma\tau)| \; .
\]

Let us represent $k$ white vertices labeled $1$ to $k$ and $k$ black vertices labeled $1$ to $k$.
We connect the white vertex $p$ with the black vertex $\tau(p)$ and the black vertex $q$ with the white vertex 
$\xi(q)$. The cycles of $\xi\tau$ are the cycles made of alternating $\xi$ and $\tau$ edges (the connected components of this graph).

If there exists a $p$ such that $ \xi \tau (p) \neq p $
we compare $|C(\xi\sigma^{-1} )| + |C(\xi \tau)|  $ with $|C(\sigma \tilde \xi^{-1} )| + |C( \tilde \xi \tau)|$ where:
\[
 \tilde \xi (q) = \begin{cases}
                  \xi(q)\; & , \quad \text{for}\; q \neq \xi^{-1}(p), \tau(p) \\ 
                   \xi(\tau(p))=  \xi \tau \xi (q) \; &, \quad \text{for}\; q = \xi^{-1}(p)  \\
                    p = \tau^{-1}(q)  \; &, \quad \text{for}\; q =\tau(p)  
                  \end{cases} \; .
\]
The number of cycles of $\xi\tau$ goes up by one:
\[
 |C(\xi \tau)|  = 1 + |C( \tilde \xi \tau)| \;,
\] 
while the number of cycles of $ \sigma \xi^{-1}$ can not decrease by more than one (looking at the graph corresponding to $\xi$ and $\sigma^{-1}$
we see that the change from $\xi$ to $\tilde \xi$ can at most collapse two connected components into one). As $|C(\xi\sigma^{-1} )|  = |C(\sigma\xi^{-1} )| $
we conclude that:
\[
 |C(\sigma \xi^{-1} )| + |C(\xi \tau)| \le |C(\sigma \tilde \xi^{-1} )| + |C( \tilde \xi \tau)| \;, 
\]
and now $\tilde \xi \tau(p)=p$. Iterating we obtain:
\[
 |C(\sigma \xi^{-1} )| + |C(\xi \tau)| \le \left(  |C(\sigma \xi^{-1} )| + |C(\xi \tau)| \right)_{\xi=\tau^{-1}} = k + |C(\sigma \tau)| \; .
\]
\end{proof}

A double application of lemma \ref{permutationinequalitylemma} leads to 
\begin{equation}
|C(\zeta\tau^{-1})|+|C(\tau\sigma^{-1})|\leq k+|C(\zeta\sigma^{-1})|\leq 2k+|C(\zeta)|-|C(\sigma)| \; ,
\end{equation}
and taking into account that $|C(\sigma)|=|\pi|$ (the number of integers in the partition $\pi$) and $|C(\zeta)|=|B(G)|$ is the number of broken faces, we arrive at: 

\begin{equation}\label{amplitudecumulantsinequality}
 \bigg|{\cal A}_{(G,T)}^{\pi}(\lambda,N)
\bigg|\leq 
\frac{2^{2k}(k!)^{2}|\lambda|^{|E(G)|}N^{|V(G)|-|E(G)|+|F(G)|-|\pi|}}{\big(\cos\frac{\arg\lambda}{2}\big)^{2|E(G)|+k}|V(G)|!(|E(G)|-|V(G)|+1)!}
 \end{equation}
In particular, for a tree we have
\begin{equation}
 \bigg|{\cal A}_{T}^{\pi}(\lambda,N)
\bigg|\leq 
\frac{2^{2k}(k!)^{2}|\lambda|^{|E(T)|}N^{2-|\pi|}}{\big(\cos\frac{\arg\lambda}{2}\big)^{2|E(T)|+k}|V(T)|!} \; .
\end{equation}
This establishes theorem \ref{treecumulants:thm}.

In order to prove theorem \ref{perturbativecumulants:thm} we apply the same algorithm as before and obtain the perturbative
series with remainder:
\begin{equation}
{\cal R}_{\pi,n}(\lambda,N)=
\sum_{(G,T)\text{ LVE graph}\atop |E(G)|= n+1\text{ and } |K(G)|=k}
{\cal A}^{\pi}_{(G,T)}(\lambda,N)
+
\sum_{T\text{ LVE tree}\atop |E(T)|\geq n+2\text{ and } |K(T)|=k}
{\cal A}_{T}^{\pi}(\lambda,N) \; ,
\end{equation}
where $K(G)$ denotes the number of cilia of $G$.
In order to bound this remainder we use eq. \eqref{amplitudecumulantsinequality} and bound separately 
the contributions of the trees and of the LVE graphs with loop edges.

The contribution of trees with  $E(T)\geq n+2$ edges is bounded by:
\begin{equation}
\bigg|\sum_{T\text{ LVE tree}\atop |E(T)|\geq n+2\text{ and } |K(T)|=k}
{\cal A}_{T}^{\pi}(\lambda,N)\bigg|\leq\sum_{n'\geq n+2}
{\cal N}(n',k)\frac{2^{2k}(k!)^{2}|\lambda|^{n'}N^{2-|\pi|}}{\big(\cos\frac{\arg\lambda}{2}\big)^{2n'+k}(n'+1)!}
\end{equation}
The number of LVE trees with $n'$ edges and $k$ cilia ${\cal N}(n',k) $ has been evaluated in lemma
\ref{coutingtrees:lem} and we get:
\begin{multline}
\bigg|\sum_{T\text{ LVE tree}\atop |E(T)|\geq n+2\text{ and } |K(T)|=k}
{\cal A}_{T}^{\pi}(\lambda,N)\bigg|\leq\frac{N^{2-|\pi|}2^{3k-1}  k! }{\big(\cos\frac{\arg\lambda}{2}\big)^{k}}
\sum_{n'\geq n+2}
\frac{(n'-1)!}{(n'+1-k)! }
\frac{2^{2n'}|\lambda|^{n'}}{\big(\cos\frac{\arg\lambda}{2}\big)^{2n'}} \; .
\end{multline}
At fixed $k$ the ratio $\frac{(n'-1)!}{(n'+1-k)!}$ is a polynomial in $n'$ so that the series on the right hand side is
convergent for $ 4 |\lambda| < \cos^{2}\frac{\arg\lambda}{2} $. It can be rewritten as:
\begin{equation}
\sum_{n'\geq n+2}
\frac{(n'-1)!}{(n'+1-k)!}
\frac{2^{2n'}|\lambda|^{n'}}{\big(\cos\frac{\arg\lambda}{2}\big)^{2n'}} =
\sum_{m\geq 0}
\frac{(m+n+1)!}{(m+n+3-k)!} \left(  \frac{ 4 |\lambda|}{  \big(\cos\frac{\arg\lambda}{2}\big)^{2} } \right)^{m+n+2} \; ,
\end{equation}
and, as $n-1\ge k$ (as the LVE trees have at most a cilium per vertex) we have $m! \ge (m+n+3-k)!$, hence:
\begin{align}
& \bigg|\sum_{T\text{ LVE tree}\atop |E(T)|\geq n+2\text{ and } |K(T)|=k}
{\cal A}_{T}^{\pi}(\lambda,N)\bigg| \leq \crcr
& \qquad \leq \frac{N^{2-|\pi|}2^{3k-1}  k! }{\big(\cos\frac{\arg\lambda}{2}\big)^{k}}
 (n+1)! \sum_{m\ge 0} \binom{m+n+1}{m}\left(  \frac{ 4 |\lambda|}{  \big(\cos\frac{\arg\lambda}{2}\big)^{2} } \right)^{m+n+2} = \crcr
& \qquad = \frac{N^{2-|\pi|}2^{3k-1}  k! }{\big(\cos\frac{\arg\lambda}{2}\big)^{k}}
 (n+1)!  \frac{
   \left(  \frac{ 4 |\lambda|}{  \big(\cos\frac{\arg\lambda}{2}\big)^{2} } \right)^{n+2}
 }{ \left( 1 -  \frac{ 4 |\lambda|}{  \big(\cos\frac{\arg\lambda}{2}\big)^{2} } \right)^{n+2} } \; .
\end{align}
 
Denoting by $n'$ the number of edges in a spanning tree (so that the graph has $n'+1$ vertices) and $n''$ the number of loop edges
of a LVE graphs with loop edges, the contribution to the rest term of these graphs is bounded by: 
\begin{multline}
 \bigg|
\sum_{(G,T)\text{ LVE graph}\atop |E(G)|= n+1\text{ and } |K(G)|=k}
{\cal A}^{\pi}_{(G,T)}(\lambda,N)
\bigg|\leq \\
\leq \sum_{n'+n''=n+1}{\cal N}(n',n'',k)
\frac{2^{2k}(k!)^{2}|\lambda|^{n'+n''}N^{2-|\pi|}}{\big(\cos\frac{\arg\lambda}{2}\big)^{2n'+2n''+k}(n'+1)!(n'')!} \; ,
\end{multline}
where ${\cal N}(n',n'',k)$ is the number of LVE graphs with $n''$ loop edges, $n'+1$ vertices and $k$ cilia.  

The following lemma is an immediate consequence of the counting of LVE graph with given number of vertices, 
cilia and loop edges performed in \cite {Razvannonperturbative}.

\begin{lemma}[Counting LVE graphs]
The number of LVE graphs with $n'+1$ vertices, $n''$ loop edges and $k$ cilia reads 
\begin{equation}
{\cal N}(n',n'',k)=\frac{(2 n'+2 n'' +k-1)!(n'+1)!}{(n'+k)!2^{n''}k!(n'+1-k)!} \; .
\end{equation}
\end{lemma}
\begin{proof}
First notice that that the number of LVE graphs with $n'+1$ vertices, $k$ cilia on a specific set of vertices and $n''$ loop edges reads \cite{Razvannonperturbative}
\begin{equation}
\frac{\big(2(n'+n'')+k-1)\big)!}{2^{n''}(n'+ k )!} \; .
\end{equation}
Then, we obtain ${\cal N}(n',n'',k)$ by counting the configurations of vertices that can carry cilia, so that
\[
{\cal N}(n',n'',k)=\frac{\big(2(n'+n'')+k-1\big)!}{2^{n''}(n'+ k )!}\times\frac{(n'+1)!}{(n'+1-k)!k!}
=
\frac{(2n'+2n''+k-1)!(n'+1)!}{(n'+k)!2^{n''}k!(n'+1-k)!} \; .
\]
\end{proof}

Using the binomial bound $\frac{1}{(n'+k)!(n'')!}\leq\frac{2^{n'+n''+k}}{(n'+n''+k)!}$ and the trivial inequality $\frac{1}{(n'+k-1)!}\leq 1$, we arrive at
\begin{equation}
 \bigg|
\sum_{(G,T)\text{ LVE graph}\atop |E(G)|= n+1}
{\cal A}^{\pi}_{(G,T)}(\lambda,N)
\bigg|\leq
\frac{2^{3k+n+1}k!|\lambda|^{n+1}N^{2-|\pi|}(2n+k+1)!}{\big(\cos\frac{\arg\lambda}{2}\big)^{2n+2+k}(n+k+1)!}
\sum_{n'+n''=n+1} 1 \; .
\end{equation}
Another use of the binomial formula shows that $\frac{(2n+k+1)!}{(n+k+1)!}\leq 2^{2n+k+1}n!$ so that
\begin{equation}
 \bigg|
\sum_{(G,T)\text{ LVE graph}\atop |E(G)|= n+1\text{ and } |K(T)|=k}
{\cal A}^{\pi}_{(G,T)}(\lambda,N)
\bigg|\leq
\frac{2^{4k+3n+2}k!|\lambda|^{n+1}N^{2-|\pi|}n!(n+2)}{\big(\cos\frac{\arg\lambda}{2}\big)^{2n+2+k}}
\end{equation}
Since $n+2<2(n+1)$,this also implies
\begin{equation}
 \bigg|
\sum_{(G,T)\text{ LVE graph}\atop |E(G)|= n+1\text{ and } |K(T)|=k}
{\cal A}^{\pi}_{(G,T)}(\lambda,N)
\bigg|\leq
\frac{2^{4k+3n+3}k!|\lambda|^{n+1}N^{2-|\pi|} (n+1)!}{\big(\cos\frac{\arg\lambda}{2}\big)^{2n+2+k}}
\end{equation}

Summing up the two bounds we obtain:
\begin{align*}
& \Big|{\cal R}_{\pi,n}(\lambda,N)\Big|
\leq \crcr
& \; \leq \frac{N^{2-|\pi|}2^{3k-1}  k! }{\big(\cos\frac{\arg\lambda}{2}\big)^{k}}
 (n+1)!  \frac{
   \left(  \frac{ 4 |\lambda|}{  \big(\cos\frac{\arg\lambda}{2}\big)^{2} } \right)^{n+2} }
   { \left( 1 -  \frac{ 4 |\lambda|}{  \big(\cos\frac{\arg\lambda}{2}\big)^{2} } \right)^{n+2} }
 +\frac{2^{4k+3n+3}k!|\lambda|^{n+1}N^{2-|\pi|} (n+1)!}{\big(\cos\frac{\arg\lambda}{2}\big)^{2n+2+k}} \crcr
& \; = N^{2-|\pi|} \left( \frac{2^{3k-1}k!}{ \big(\cos\frac{\arg\lambda}{2}\big)^{k} } \right) (n+1)!   
\left(  \frac{ 4 |\lambda|}{  \big(\cos\frac{\arg\lambda}{2}\big)^{2}} \right)^{n+1} 
\left(
 \frac{    \frac{ 4 |\lambda|}{  \big(\cos\frac{\arg\lambda}{2}\big)^{2} }   }
   { \left( 1 -  \frac{ 4 |\lambda|}{  \big(\cos\frac{\arg\lambda}{2}\big)^{2} } \right)^{n+2} }
   + 2^{k+n+2 }
\right)\; .
\end{align*}

This establishes theorem \ref{perturbativecumulants:thm}. 

\subsection{Topological expansion for the cumulants (proof of Theorem \ref{topologicalcumulants:thm})}

The proof of theorem \ref{topologicalcumulants:thm}
proceeds along the same lines as that of theorem \ref{perturbativecumulants:thm}, except that the perturbative expansion involves 
contributions of all graphs up to genus $g$ and the remainder contains graphs of genus $g+1$. As before, the main idea is to express 
the contribution of Feynman graphs and LVE in terms of trace invariants.

Starting with eq. \eqref{topologicalexpansion:eq}, we collect terms homogeneous of degree $k$ in $JJ^{\dagger}$. The first term is obviously 
written as a sum of trace invariants
\begin{multline}
\sum_{G\text{ ribbon graph with $k$ cilia}\atop\text{ broken faces corresponding to $\pi$ and $g(G)\leq g$}}\frac{(-\lambda)^{|E(G)|}N^{\chi(G)}}{|\text{Aut}(G)|}
\prod_{f\text{ broken face}}\big(JJ^{\dagger}\big)^{c(f)}=\\
\sum_{G\text{ ribbon graph with $k$ cilia }\atop\text{ broken faces corresponding to $\pi$ and $g(G)\leq g$}}\frac{(-\lambda)^{|E(G)|}N^{\chi(G)}}{|\text{Aut}(G)|}
\Tr_{\pi}(JJ^{\dagger}) \; .
\end{multline}
After derivation with respect to the sources, it yields a contribution to $K_{\pi,g}(\lambda,n)$ of the form
\begin{equation}
\sum_{G\text{ ribbon graph with $k$ cilia}\atop\text{ broken faces corresponding to $\pi$ and $g(G)\leq g$}}\frac{(-\lambda)^{|E(G)|}N^{\chi(G)}}{|\text{Aut}(G)|} \;, 
\end{equation}
which is just the sum over Feynman graph of genus less than $g$ and broken faces corresponding to $\pi$.
Recall that the number of graphs of genus $g$ with $n$ edges is bounded as in \eqref{ciliabound:eq}:
\begin{equation}
\widetilde{\cal N}(g,n,k)\leq
C''_{g}12^{n} n^{\frac{5}{2}(g-1)}
\frac{(2n)!}{k!(2n-k)!} \;.
\label{ciliacbound2:eq}
\end{equation}
Accordingly, the series $\sum_{n}\widetilde{\cal N}(g,n,k)z^{k}$ converges for $|z|<\frac{1}{12}$.

Consider now the remainder $\widetilde{R}_{\pi,g}(\lambda,N)$, containing only genus $g+1$ LVE graphs and express it in terms 
of trace invariants: 
\begin{equation}
\widetilde{{\cal R}}_{\pi,g}(\lambda,N)=
\sum_{(G,T)\text{ LVE graphs with broken faces corresponding to $\pi$}\atop
g(G)=g+1 \text{ and } g(G-e_{|L(G,T)|})=g}
{\cal A}^{\pi}_{(G,T)}(\lambda, N) \; .
\end{equation}

We bound each LVE graph as in \eqref{amplitudecumulantsinequality}. Since this bound only depends on the
graph $G$ and not on the choice of the spanning tree $T$, we use lemma \ref{Heppsectors} to get
\begin{equation}
\big|\widetilde{R}_{\pi,g}(\lambda,N)\big|
\leq \sum_{n=2(g+1)}^{\infty}
\widetilde{\cal N}(g+1,n,k)\frac{
2^{2k}(k!)^{2}|\lambda|^{|n|}N^{2-2(g+1)-|\pi|}}{\big(\cos\frac{\arg\lambda}{2}\big)^{2n+k}}
\end{equation}
hence
\begin{align*}
& \big|\widetilde{R}_{\pi,g}(\lambda,N)\big| \leq \crcr
& \leq N^{2-2(g+1)-|\pi|} \frac{2^{3k} k! }{ \big(\cos\frac{\arg\lambda}{2}\big)^{k} } C''_{g+1}
\sum_{n=2(g+1)}^{\infty}  n^{\frac{5}{2}(g-1) + k} \left( \frac{12|\lambda|}{\big(\cos\frac{\arg\lambda}{2}\big)^{2} } \right)^n \crcr
& = N^{2-2(g+1)-|\pi|} \frac{2^{3k} k! }{ \big(\cos\frac{\arg\lambda}{2}\big)^{k} } C''_{g+1}
\left( \frac{12|\lambda|}{\big(\cos\frac{\arg\lambda}{2}\big)^{2} } \right)^{2g+2} \crcr
& \qquad \qquad \times \sum_{m\ge 0} (m+g+1)^{ \frac{5}{2}(g-1) + k  } \left( \frac{12|\lambda|}{\big(\cos\frac{\arg\lambda}{2}\big)^{2} } \right)^{m} \; .
\end{align*}
Bounding:
\[
  (m+g+1)^{ \frac{5}{2}(g-1) + k  }  \le (m+g+1)^{3g+k} \le \frac{ (m+4g+k+1)! }{ m!  } \;,
\]
we obtain the rough bound:
\begin{align*}
 & \big|\widetilde{R}_{\pi,g}(\lambda,N)\big| \leq \crcr
 &\leq  N^{2-2(g+1)-|\pi|} \frac{2^{3k} k! }{ \big(\cos\frac{\arg\lambda}{2}\big)^{k} } C''_{g+1}
\left( \frac{12|\lambda|}{\big(\cos\frac{\arg\lambda}{2}\big)^{2} } \right)^{2g+2} 
\frac{(4g+k+1)! }{ \left( 1 -    \frac{12|\lambda|}{\big(\cos\frac{\arg\lambda}{2}\big)^{2} } \right)^{4g+k} } \;. 
\end{align*}

This achieves the proof of theorem \ref{topologicalcumulants:thm}.

\section*{Acknowledgements:}

Both authors thank V. Rivasseau for very fruitful discussions and the Erwin Schr\"odinger Institute
for hospitality during the program "Combinatorics, Geometry and Physics". T.K. also thanks the Centre de Physique 
Th\'eorique at Ecole Polytechnique for hospitality and University Paris-Nord for support.

\appendix

\section{Schwinger-Dyson equations for the intermediate field}

\label{SDappendix}

In this appendix, we derive the explicit formula for the order 2 cumulant in the large $N$ limit using the Schwinger-Dyson equation for the intermediate field.  
This allows us to express it as an explicit power series collecting all the planar graphs, with two external legs in the matrix model formulation or, equivalently,
with one cilium in the intermediate field representation. This result is classical and goes back to Koplik, Neveu and Nussinov \cite{KPN} which treated the $\phi^{4}$ 
interaction without recourse to the intermediate field. 

Let us collect all planar graphs contributing to the order 2 cumulant in the power series
\begin{equation}
G(\lambda,q)=\sum_{m,n}G_{m,n}q^{m}\lambda^{n}
\end{equation}
with $G_{m,n}$ the number of connect planar graphs with $n$ vertices (intermediate field edges or equivalently, matrix vertices $\text{Tr}(MM^{\dagger}MM^{\dagger})$ )
and $m$ boundary matrix lines. The indeterminate $q$ only appears in the formulation of the Schwinger-Dyson equation and the order 2 cumulant is recovered in the limit $q\rightarrow 1$. 

\begin{figure}[htb]

\[
\parbox{3cm}{\includegraphics[width=3cm]{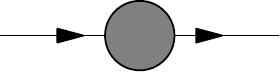}}
\quad
=
\quad
\parbox{1.5cm}{\includegraphics[width=1.5cm]{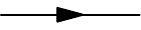}}
\quad+\quad
\begin{minipage}{4cm}{\vskip2cm\includegraphics[width=4cm]{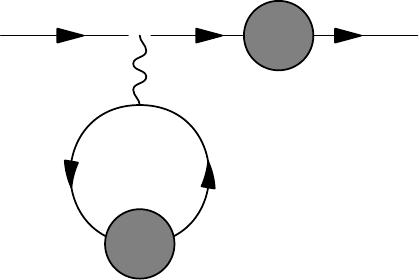}}\end{minipage}
\quad+\quad\begin{minipage}{4cm}{\vskip0.5cm\parbox{4cm}{\includegraphics[width=4cm]{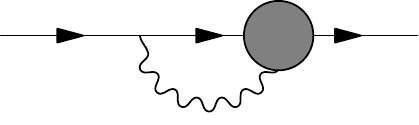}}}\end{minipage}
\]
\caption{Schwinger-Dyson equations in the intermediate field representation}
\label{SDA:fig}
\end{figure}

It obeys the planar Schwinger-Dyson equation
\begin{equation}
G(\lambda,q)=q+\lambda q G^{2}(\lambda,q)+\lambda q^{2}\frac{G(\lambda,q)-G(\lambda,1)}{q-1}\label{planarSDE}
\end{equation}
Graphically, this equation can be derived as follows. Starting from the incoming $M$ line, there are three possibilities, see figure \ref{SDA:fig}.

\begin{itemize}
\item We do not meet an interaction with the intermediate field so that the graph reduces to a single matrix line (first term in figure \ref{SDA:fig} and \eqref{planarSDE}).

\item The removal of the first $AMM^{\dagger}$ vertex we encounter disconnects the graph (second term in figure\ref{SDA:fig} and \eqref{planarSDE}).

\item The graph remains connected after the removal of the first vertex encountered (third term in figure \ref{SDA:fig} and \eqref{planarSDE}). In this case, we have 
to keep track of the various possibilities of attaching the intermediate field line on the external boundary of the graph. This is the origin of the $q$-derivative term 
\begin{equation}
\frac{G(\lambda,q)-G(\lambda,1)}{q-1}=\sum_{m\geq 1,n\geq 0}G_{m,n}\Big(\sum_{k=0}^{m-1}q^{k}\Big)\lambda^{n}
\end{equation}
since the insertion of the $\A$ line encloses $k$ $\M$-lines for $ 0\leq k\leq m-1$.

\end{itemize}

At lowest order in $n$, the explicit expressions of $G_{n}(q)=\sum_{m}G_{m,n}q^{n}$ read
\begin{align}
G_{0}(q)&=q\\
G_{1}(q)&=q^{2}[2]_{q}\\
G_{2}(q)&=2q^{4}[2]_{q}+q^{2}([2]_{q}+[3]_{q})\\
G_{3}(q)&=q^{5}([2]_{q})^{2}+4q^{5}[2]_{q}+2q^{3}([2]_{q} +[3]_{q})\\&+2([4]_{q}+[5]_{q})+(2[2]_{q}+2[3]_{q}+[4]_{q})
\end{align}
with $[n]_{q}=\frac{q^{n}-1}{q-1}$

The planar Schwinger-Dyson equation \eqref{planarSDE}  $G(1,\lambda)$ is quadratic in $G(q)$ and can be used to determine $G(q)$ in terms of the variables $q$ and 
$\lambda$ and of $G(1,\lambda)$ treated as an independent variable
\begin{equation}
G(q,\lambda)=\frac{q\!-\!1\!-\!\lambda q^{2}-\sqrt{(q\!-\!1\!-\!\lambda q^{2})^{2}-4\lambda q(q\!-\!1)[q(q\!-\!1)\!-\!\lambda q^{2}G(1,\lambda)]}}{\lambda q(q\!-\!1)}
\end{equation}
where we have retained the solution with a well defined limit at $\lambda=0$. 

In order for $G(q)$ to be analytic in $q$ and $\lambda$, $G(1,\lambda)$ which is itself a series in $\lambda$, must be such that the polynomial under the square 
root has a double root in $q$. This polynomial reads
\begin{equation}
(\lambda^{2}-4\lambda+4\lambda^{2}G(1,\lambda))q^{4}
+(6\lambda-4\lambda^{2}G(1,\lambda))q^{3}+(1-2\lambda)q^{2}-2q+1
\end{equation}
and its discriminant factorizes as
\begin{equation}
\lambda^{4}(1-\lambda G(1,\lambda))^{2}(1-G(1,\lambda)-16\lambda+18
\lambda G(1,\lambda)-27\lambda^{2}G^{2}(1,\lambda))
\end{equation}
Discarding the solution $G(1,\lambda)=\frac{1}{\lambda}$, the solution of the quadratic part yields the planar contribution to the order 2 cumulant
\begin{equation}
G(1,\lambda)=
\frac{-1+18\lambda-(1-12\lambda)^{3/2}}{54\lambda^{2}}
\end{equation}
Its expansion as a power series in $\lambda$ reads
\begin{equation}
G(1,\lambda)=\sum_{n}\frac{2\cdot3^{n}}{n+2}C_{n}\lambda^{n}\qquad\text{with}\quad C_{n}=\frac{(2n)!}{n!^{2}(n+1)}
\end{equation}
This reproduces the counting of planar ribbon graphs with one cilium, or equivalently, rooted, bipartite quadrangulations.

\section{The BKAR forest formula}\label{BKARsec}

In this appendix, we briefly review the Brydges-Kennedy-Abdesselam-Rivasseau (BKAR) forest formula \cite{BKAR}
which allows us to expand $\log {\cal Z}[J,J^{\dagger},\lambda,N]$ as a sum over trees.
 
Let $\phi$ be a function of ${\mathbb{ R} }^{\frac{n(n-1)}{2}}$ whose arguments $u_{ij}$ are associated to the edges of the
complete graph on $n$ vertices labeled $\left\{1,2,\dots,n\right\}$. The following theorem yields an expansion of $\phi(1,\dots,1)$ as a 
sum over forests with $n$ labeled vertices, Recall that a forest is a subset of edges of the compete graph that does not contain any cycle. 

For every forest $F$, let us denote (if it exists) ${P}_{i\leftrightarrow j}^{F}$ the unique path in the forest $F$ joining the vertices 
$i$ and $j$.

\begin{theorem}[Brydges-Kennedy-Abdesselam-Rivasseau]
Let $\phi: \,{ \mathbb{ R} }^{\frac{n(n-1)}{2}}\rightarrow {\mathbb C}$ be a smooth, sufficiently derivable function. Then:
\begin{equation}
\phi(1,\dots,1)=\sum_{F\text{ forest}}\int_{0}^{1}\prod_{(i,j)\in{F}}du_{ij} \; 
 \left( \frac{\partial^{|E(F)|} \phi}{\prod_{(i,j)\in{F}} \partial x_{ij}} \right) \big(v^{F}_{ij}\big) \; ,
\label{BKARformula}
\end{equation}
where $v^F_{ij}$ is given by: 
\begin{equation}
v^{F}_{ij}=\left\{\begin{array}{ccl}
\inf_{(k,l)\in{ P}_{i\leftrightarrow j}^{{F}}} u_{kl}&\text{if}&  { P}_{i\leftrightarrow j}^{{F}} \; \text{exists} \\
0&\text{if}&  { P}_{i\leftrightarrow j}^{{F}} \; \text{does not exist} \
\end{array}\right. \; ,
\end{equation}
and $|E(F)|$ is the number of edges in the forest $F$. 
\end{theorem}

This theorem is a broad generalization of the fundamental theorem of calculus, to which it reduces when 
$n=2$. Indeed, in this case there are two forests on the complete graph with 2 vertices (thus with a single edge, 
see figure \ref{forest2:fig}) and we get
\begin{equation}
\phi(1)=\phi(0)+\int_{0}^1 dt_{12} \; \left( \frac{\partial \phi}{\partial x_{12}} \right) (t_{12})
\end{equation}
The first term corresponds to the empty forest ($|E(F)|=0$) and the second one to the full forest ($|E(F)|=1$).
\begin{figure}[htb]
\begin{center}
\begin{tabular}{ccc}
\includegraphics[width=1.5cm]{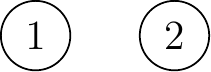}&,&
\includegraphics[width=1.5cm]{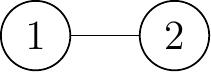}
\end{tabular}
\end{center}
\caption{The two forests built on two vertices}
\label{forest2:fig}
\end{figure}

For $n=3$, we have seven forests (see figure \ref{forest3:fig}) and we get:
\begin{align*}
\phi(1,1,1) & = \phi(0,0,0) + \int_{[0,1]}dt_{12} \left( \frac{\partial \phi}{\partial x_{12}} \right) (t_{12},0,0) \crcr
& +\int_{[0,1]}dt_{23} \left( \frac{\partial \phi}{\partial x_{23}} \right) (0,t_{23},0) 
 +\int_{[0,1]}dt_{13} \left( \frac{\partial \phi}{\partial x_{13}} \right) (0,0,t_{13}) \crcr
& + \int_{[0,1]^{2}} dt_{12}dt_{23}
 \left( \frac{\partial ^{2}\phi}{\partial x_{12}\partial x_{23}} \right) (t_{12},t_{23},\inf(t_{12},t_{23})) 
 \crcr
& + \int_{[0,1]^{2}}  dt_{12}dt_{13}
\left( \frac{\partial ^{2}\phi}{\partial x_{12}\partial x_{13}} \right) (t_{12},\inf(t_{12},t_{13}),t_{13}) \crcr
& + \int_{[0,1]^{2}}  dt_{23}dt_{13}
\left( \frac{\partial ^{2}\phi}{\partial x_{23}\partial x_{13}} \right) (\inf(t_{23},t_{13}),t_{23},t_{13})\; . \nonumber
 \end{align*}

\begin{figure}[htb]
\begin{center}
\begin{tabular}{cccccccccccccc}
\includegraphics[width=1.5cm]{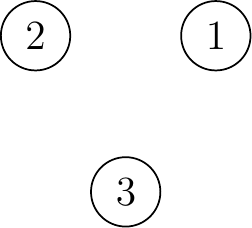}&,&
\includegraphics[width=1.5cm]{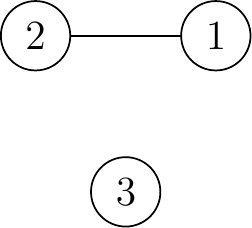}&,&
\includegraphics[width=1.5cm]{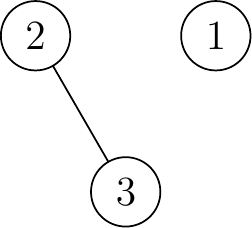}&,&
\includegraphics[width=1.5cm]{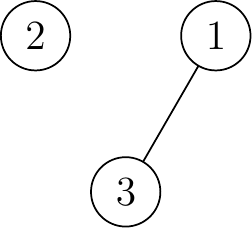}&,&
\includegraphics[width=1.5cm]{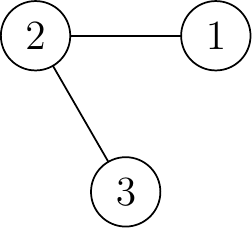}&,&
\includegraphics[width=1.5cm]{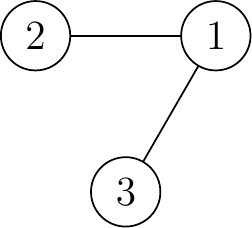}&,&
\includegraphics[width=1.5cm]{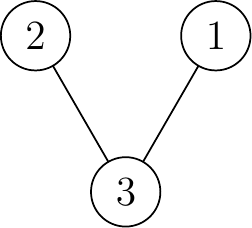}&
\end{tabular}
\end{center}
\caption{The seven forests built on three vertices}
\label{forest3:fig}
\end{figure}

The first term corresponds to the empty forest, the next three to the forests with one edge and the last three to the forests with two edges.

In quantum field theory, the main interest of this formula lies in the fact that it provides an expansion for the partition function 
with sources ${\cal Z}[J]$ as a sum over forests. Its logarithm is then readily computed as a \emph{convergent} 
sum over trees.

\section{Some examples of LVE graphs and their amplitudes}\label{LVEexamples:app}

Here we illustrate how the LVE graph amplitude \eqref{LVEamplitude} is computed on a few examples. 
We use a double line representation for the intermediate field instead of a wavy 
line, $\parbox{2cm}{\includegraphics[width=2cm]{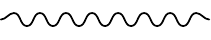}}\rightarrow\parbox{2cm}{\includegraphics[width=2cm]{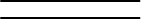}}$

\begin{figure}[htb]
\begin{center}
\includegraphics[width=6cm]{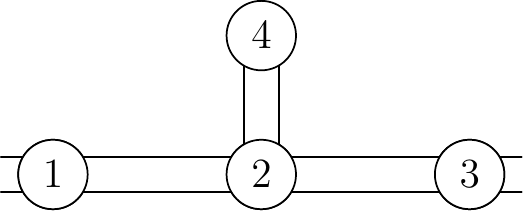}
\end{center}
\caption{A LVE tree}
\label{LVEtree2:fig}
\end{figure}

For the tree in figure \ref{LVEtree2:fig}, the amplitude reads:
\begin{multline}
{\cal A}_{\includegraphics[width=1cm]{LVEtree2.pdf}}=\frac{N(-\lambda)^{4}}{3!}\int_{0}dt_{12}dt_{23}dt_{24}
\int d\mu_{C}(A)\nonumber\\
\Tr\bigg[
\Big(1-\text{i}\sqrt{\frac{\lambda}{N}}A_{3}\Big)^{-1}
\Big(1-\text{i}\sqrt{\frac{\lambda}{N}}A_{2}\Big)^{-1}
\Big(1-\text{i}\sqrt{\frac{\lambda}{N}}A_{4}\Big)^{-1}
\Big(1-\text{i}\sqrt{\frac{\lambda}{N}}A_{2}\Big)^{-1}\nonumber\\
\Big(1-\text{i}\sqrt{\frac{\lambda}{N}}A_{1}\Big)^{-1}
JJ^{\dagger}
\Big(1-\text{i}\sqrt{\frac{\lambda}{N}}A_{1}\Big)^{-1}
\Big(1-\text{i}\sqrt{\frac{\lambda}{N}}A_{2}\Big)^{-1}
\Big(1-\text{i}\sqrt{\frac{\lambda}{N}}A_{3}\Big)^{-1}
JJ^{\dagger}
\bigg] \; ,
\end{multline}
with covariance matrix
\begin{equation}
C=\begin{pmatrix}
1&t_{12}&\inf(t_{12},t_{23})&\inf(t_{12},t_{24})\\
t_{12}&1&t_{23}&t_{24}\\
\inf(t_{12},t_{23})&t_{23}&1&\inf(t_{23},t_{24})\\
\inf(t_{12},t_{24})&t_{24}&\inf(t_{23},t_{24})&1
\end{pmatrix} \;.
\end{equation}

\begin{figure}[htb]
\begin{center}
\includegraphics[width=7cm]{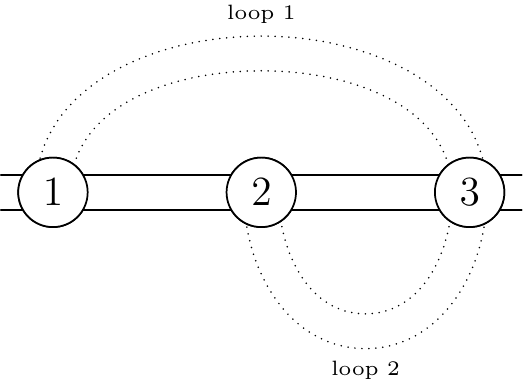}
\end{center}
\caption{A planar LVE graph}
\label{LVEgraph2:fig}
\end{figure}

For the  planar LVE graph in figure \ref{LVEgraph2:fig}, the amplitude reads:
\begin{equation}
\hskip-0.5cm{\cal A}_{\includegraphics[width=1cm]{LVEgraph2.pdf}}=\frac{N^{-1}(-\lambda)^{4}}{3!}
\int_{0}^{1}ds_{1}\int_{0}^{s_{1}}ds_{2}
\int_{0}dt_{12}dt_{23}dt_{24}
\inf(t_{12},t_{23})t_{23}
\int d\mu_{s_2C}(A)
\end{equation}
\vskip-0.7cm
\begin{multline*}
\hskip-0.5cm\textstyle{\Tr\bigg[
\Big(1-\text{i}\sqrt{\frac{ \lambda}{N}}A_{3}\Big)^{-1}
\Big(1-\text{i}\sqrt{\frac{ \lambda}{N}}A_{1}\Big)^{-1}
JJ^{\dagger}
\Big(1-\text{i}\sqrt{\frac{ \lambda}{N}}A_{1}\Big)^{-1}
\Big(1-\text{i}\sqrt{\frac{ \lambda}{N}}A_{2}\Big)^{-1}}\\
\textstyle{\Big(1-\text{i}\sqrt{\frac{ \lambda}{N}}A_{3}\Big)^{-1}
JJ^{\dagger}
\bigg]
\Tr\bigg[
\Big(1-\text{i}\sqrt{\frac{ \lambda}{N}}A_{1}\Big)^{-1}
\Big(1-\text{i}\sqrt{\frac{ \lambda}{N}}A_{2}\Big)^{-1}}\\
\textstyle{\Big(1-\text{i}\sqrt{\frac{ \lambda}{N}}A_{3}\Big)^{-1}
\bigg]
\Tr\bigg[
\Big(1-\text{i}\sqrt{\frac{ \lambda}{N}}A_{2}\Big)^{-1}
\Big(1-\text{i}\sqrt{\frac{ \lambda}{N}}A_{3}\Big)^{-1}
\bigg]} \; ,
\end{multline*}
with covariance matrix:
\begin{equation}
\textstyle{s_2C=s_2\begin{pmatrix}
1&t_{12}&\inf(t_{12},t_{23})\\
t_{12}&1&t_{23}\\
\inf(t_{12},t_{23})&t_{23}&1
\end{pmatrix}} \;. 
\end{equation}

\begin{figure}[htb]
\begin{center}
\includegraphics[width=7cm]{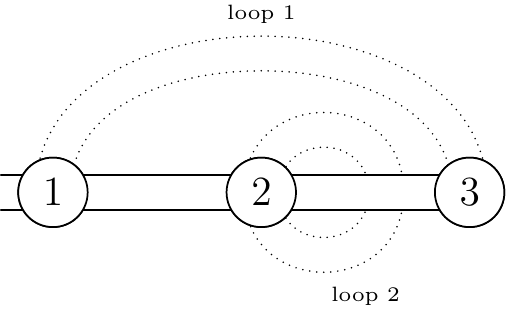}
\end{center}
\caption{A non planar LVE graph}
\label{LVEgraph:fig}
\end{figure}

For the  non planar LVE graph in figure \ref{LVEgraph:fig}, the amplitude reads:
\begin{multline}
{\cal A}_{\includegraphics[width=1cm]{LVEgraph.pdf}}=\frac{N^{-1}(-\lambda)^{4}}{3!}
\int_{0}^{1}ds_{1}\int_{0}^{s_{1}}ds_{2}
\int_{0}dt_{12}dt_{23}dt_{24}
\inf(t_{12},t_{23})
\int d\mu_{s_2C}(A)\nonumber\\
\Tr\bigg[
\Big(1-\text{i}\sqrt{\frac{ \lambda}{N}}A_{1}\Big)^{-1}
\Big(1-\text{i}\sqrt{\frac{ \lambda}{N}}A_{2}\Big)^{-1}
\Big(1-\text{i}\sqrt{\frac{ \lambda}{N}}A_{2}\Big)^{-1}
\Big(1-\text{i}\sqrt{\frac{ \lambda}{N}}A_{1}\Big)^{-1}\nonumber\\
\Big(1-\text{i}\sqrt{\frac{ \lambda}{N}}A_{3}\Big)^{-1}
\Big(1-\text{i}\sqrt{\frac{ \lambda}{N}}A_{2}\Big)^{-1}
\Big(1-\text{i}\sqrt{\frac{ \lambda}{N}}A_{2}\Big)^{-1}\\
\nonumber
\Big(1-\text{i}\sqrt{\frac{ \lambda}{N}}A_{3}\Big)^{-1}
\Big(1-\text{i}\sqrt{\frac{ \lambda}{N}}A_{1}\Big)^{-1}
JJ^{\dagger}
\bigg] \; ,
\end{multline}
with covariance matrix
\begin{equation}
s_2C=s_2\begin{pmatrix}
1&t_{12}&\inf(t_{12},t_{23})\\
t_{12}&1&t_{23}\\
\inf(t_{12},t_{23})&t_{23}&1
\end{pmatrix} \; .
\end{equation}

\section{Analyticity domain for the vector model}

In the case of the vector model, the cardioid can be extended to reach the real negative axis. This result is not new, see \cite{Billionnet}. 
An argument similar to ours, based on the LVE and contour rotation can also be found in \cite{advanced}.

The cumulants of the vector model are defined by the generating ${\cal W}[J,J^{\dagger}]=
\log {\cal Z}[J,J^{\dagger}]$ with
\begin{align}
{\cal Z}[J,J^{\dagger}]=\int d\Phi\exp-\bigg\{
\Phi^{\dagger}\Phi+\frac{\lambda}{2N}\big(\Phi^{\dagger}\Phi\big)^{2}+\sqrt{N}\Phi^{\dagger}J+\sqrt{N} J^{\dagger}\Phi\bigg\}\label{Phiintegral}
\end{align}
where $\Phi\in{ \mathbb{ C} }^{N}$ is a $N$ component vector and $\Phi^{\dagger}\Phi=|\Phi_{1}|^{2}+\cdots+|\Phi_{N}|^{2}$. As for the matrix model, 
the integral is normalized such that ${\cal Z}[J,J^{\dagger}]=1$ at $\lambda=0$. The sources $J$ and $J^{\dagger}$ are also $N$ component complex vectors.

The vector model admits a intermediate field representation based on
\begin{equation}
\exp-\frac{\lambda }{2N}\big(\Phi^{\dagger}\Phi)^{2}
=\int dA \exp-\bigg\{\frac{1}{2}A^{2}-\mathrm{i}\sqrt{\frac{\lambda}{N}}\,A\Phi^{\dagger}\Phi\bigg\}
\end{equation}
where the integral is over a real scalar $A$. Thus, the Gau\ss ian integral over $\Phi$ can be performed
\begin{align}
{\cal Z}[J,J^{\dagger}]&=\int d\Phi dA
\exp-\bigg\{\frac{1}{2}A^{2}+\bigg[ \Phi^{\dagger}\bigg(1-\mathrm{i}\sqrt{\frac{\lambda}{N}}A\bigg)\Phi\bigg]
+\sqrt{N}\Phi^{\dagger}J+\sqrt{N}J^{\dagger}\Phi\bigg\}\\
&=
\int dA
\exp-\bigg\{\frac{1}{2}A^{2}
+N\log\bigg(1-\mathrm{i}\sqrt{\frac{\lambda}{N}}A\bigg)+NJ^{\dagger}J\bigg(1-\mathrm{i}\sqrt{\frac{\lambda}{N}}A
\bigg)^{-1}\bigg\}
\end{align}
Let us notice two differences with respect to the matrix integral eq. \eqref{Aintegral}: There is no power of $N$ in front of the resolvent and $A$ 
and $J^{\dagger}J$ are scalars, so that they can be commuted. However, the perturbative expansion of the vector model in the intermediate field also
involves ribbon graphs, even if all quantities are scalars. This is so because the interaction vertices are based on resolvents, which have a cyclic ordering of their half edges.

Then, we  perform the loop vertex expansion  and expand $\log {\cal Z}[J,J^{\dagger}]$ over ciliated ribbon trees 
\begin{equation}
\log {\cal Z}[J,J^{\dagger}]=
\sum_{T\,\text{tree}}\frac{N(-\lambda)^{(|E(T)|)}(J^{\dagger}J)^{k}}{|V(T)|!}\int dt
\int d\mu_{C_{T}}(A)\prod_{i=1}^{n}\bigg(1-\mathrm{i}\sqrt{\frac{\lambda}{N}}A_{i}\bigg)^{-l_{i}}
\end{equation} 
$n$ is the number of vertices of $T$, $k$ the number of cilia and $l_{i}$ the number of corners attached to vertex $i$. This expression has the same
domain of convergence as the matrix integral, stated in theorem \ref{treeexpansion}.

However, the domain of analyticity can be enlarged. To proceed, let us write the powers of the resolvent as
\begin{equation}
\bigg(1-\mathrm{i}\sqrt{\frac{\lambda}{N}}A\bigg)^{-l}
=\frac{1}{\Gamma(l)(\sqrt{\lambda})^{l}}
\int_{0}^{\infty}d\alpha\,\alpha^{l-1}\exp-\alpha\Big\{{\frac{1}{\sqrt{\lambda}}-\frac{\text{i}}{\sqrt{N}}}A\Big\}
\end{equation}
Inserting this representation in the integral over replicas yields
\begin{align}
\int d\mu_{C_{T}}(A)\prod_{i=1}^{n}\bigg(1-\mathrm{i}\sqrt{\frac{\lambda}{N}}A_{i}\bigg)^{-l_{i}}\!\!\!\!\!\!\!=
\!\!\!\int_{0}^{\infty}\!\!\!{\textstyle  \prod_{i}d\alpha_{i}}\frac{\prod_{i}\alpha_{i} ^{l_{i}-1}}{(\sqrt{\lambda})^{2n-2+k}\prod_{i}\Gamma(l_{i})}
\int d\mu_{C_{T}}(A)
\exp-\sum_{i}\Big\{{\frac{\alpha_{i}}{\sqrt{\lambda}}-\frac{\text{i}\alpha_{i}A_{i}}{\sqrt{N}}}\Big\}
\end{align}
The main simplification in the vector model case is that the integral over the replicas is Gau\ss ian and can be performed explicitly 
\begin{align}
\int d\mu_{C_{T}}(A)
\exp-\sum_{i}\Big\{{\frac{\alpha_{i}}{\sqrt{\lambda}}-\frac{\text{i}\alpha_{i}A_{i}}{\sqrt{N}}}\Big\}
=\exp-\Big\{\sum_{i}{\frac{\alpha_{i}}{\sqrt{\lambda}}+\sum_{i,j}\frac{C_{ij}\alpha_{i}\alpha_{j}}{2N}}\Big\}
\end{align} 
Therefore, the loop vertex expansion of the vector model reads
\begin{equation}
\log {\cal Z}[J,J^{\dagger}]=
\sum_{T}\frac{N(-1)^{(n-1)}(J^{\dagger}J)^{k}}{n!(\sqrt{\lambda})^{k}\prod_{i}\Gamma(l_{i})}\int dt\int d\alpha
\prod_{i}\alpha_{i} ^{l_{i}-1}\exp-\Big\{\sum_{i}{\frac{\alpha_{i}}{\sqrt{\lambda}}+\sum_{i,j}\frac{C_{ij}\alpha_{i}\alpha_{j}}{2N}}\Big\}
\end{equation} 
where we recall that the sum runs over ribbon trees with $n$ labeled vertices and $k$ cilia and that $l_{i}$ is the number of corners attached to vertex $i$.

Next, we write $\lambda=\rho\mathrm{e}^{\mathrm{i}\theta}$and rotate the $\alpha_{i}$ integrations by an angle $\frac{\phi}{2}$ in their  complex 
planes, $\alpha_{i}\rightarrow\mathrm{e}^{\mathrm{i}\frac{\phi}{2}}\alpha_{i}$. This is possible since the integrand is holomorphic in all the $\alpha_{i}$ 
but requires that the integrand goes to 0 on the arcs $\alpha_{i}=R\mathrm{e}^{\mathrm{i}\chi}$ with $\chi\in[0,\frac{\phi}{2}]$ and $R\rightarrow +\infty$. 
This last condition imposes that the real part of the argument of the exponential
be always positive. A sufficient condition is to impose that the linear and quadratic terms are  separately positive, $\cos(\frac{\theta-\phi}{2})>0$ and $\cos\phi>0$. 
After rotation of the contour, the generating function of the cumulants is expanded as
\begin{multline}
\log {\cal Z}[J,J^{\dagger}]=\mathrm{e}^{\mathrm{i}\phi\frac{(2n-2+k)}{2}}
\sum_{T}\frac{N(-1)^{(n-1)}(J^{\dagger}J)^{k}}{n!(\sqrt{\lambda})^{k}\prod_{i}\Gamma(l_{i})} \\
 \times \int dt\int d\alpha
\prod_{i}\alpha_{i} ^{l_{i}-1}\exp-\Big\{\mathrm{e}^{\frac{\mathrm{i}}{2}(\phi-\theta)}\sum_{i}{\frac{\alpha_{i}}{\sqrt{\rho}}+\mathrm{e}^{\mathrm{i}\phi}\sum_{i,j}\frac{C_{ij}\alpha_{i}\alpha_{j}}{2N}}\Big\} \; .
\end{multline} 
Let us emphasize that none of the terms in the sum over trees depend on $\phi$. We may therefore conveniently choose $\phi$ to enlarge the domain of analyticity. 
To this aim, let us bound the exponential as
\begin{align}
\bigg|\exp-\Big\{\mathrm{e}^{\frac{\mathrm{i}}{2}(\phi-\theta)}\sum_{i}{\frac{\alpha_{i}}{\sqrt{\rho}}+\mathrm{e}^{\mathrm{i}\phi}\sum_{i,j}\frac{C_{ij}\alpha_{i}\alpha_{j}}{2N}}\Big\}\bigg|&=
\exp-\Big\{\cos{\textstyle \frac{\phi-\theta}{2}}\sum_{i}{\frac{\alpha_{i}}{\sqrt{\rho}}+\cos\phi\sum_{i,j}\frac{C_{ij}\alpha_{i}\alpha_{j}}{2N}}\Big\}\nonumber\\
&\leq\exp-\Big\{\cos{\textstyle \frac{\phi-\theta}{2}}\sum_{i}\frac{\alpha_{i}}{\sqrt{\rho}}\Big\}
\end{align}
since the covariance is a positive matrix. Now we perform the integral over the Schwinger parameters $\alpha_{i}$
\begin{align}
\big|\log {\cal Z}[J,J^{\dagger}]\big|\leq
\sum_{T}\frac{N|\lambda|^{(n-1)}(J^{\dagger}J)^{k}}{n! \cos^{2n-2+k}\frac{\theta-\phi}{2}}
\end{align} 
Since we have to stay away from the critical half line $\theta=\pm\pi$, the best bound (given the conditions imposed by the positivity of the argument of 
the exponential in the contour rotation) is obtained for 
$\phi=\theta$ for $\theta\in[-\frac{\pi}{2},\frac{\pi}{2}]$
$\phi=\frac{\pi}{2}$ for $\theta\in[0,\frac{\pi}{2}]$ and $\phi=-\frac{\pi}{2}$ for $\theta\in[-\frac{\pi}{2},0]$. Therefore, the generating function is bounded as
\begin{align}
\big|\log {\cal Z}[J,J^{\dagger}]\big|&\leq&
\sum_{T}\frac{N|\lambda|^{(n-1)}(J^{\dagger}J)^{k}}{n!} 
\qquad&\text{for}\quad\text{ $\theta\in[-\frac{\pi}{2},\frac{\pi}{2}]$}\\
\\
\big|\log {\cal Z}[J,J^{\dagger}]\big|&\leq&
\sum_{T}\frac{N|\lambda|^{(n-1)}(J^{\dagger}J)^{k}}{n! \Big[\cos\big(\frac{|\theta|}{2}\!-\!\frac{\pi}{4}\big)\Big]^{2n-2+k}}
\qquad&\text{for}\quad\text{ $\theta\in[-\pi,-\frac{\pi}{2}]\cup[\frac{\pi}{2},\pi]$}\\
\end{align} 
Finally, an argument similar to the one presented in section \ref{LVEsec} and leading to theorem \ref{treeexpansion} 
establishes the following analyticity theorem for the vector model.
 
 \begin{theorem}[Constructive theorem for the vector model]
Cumulants are analytic functions of $\lambda$  inside the curve
\begin{equation}
{\cal C}'=\Big\{\rho\mathrm{e}^{\mathrm{i}\theta}\text{ with $4\rho<1$ for $\theta\in[\frac{-\pi}{2},\frac{\pi}{2}]$ and 
$4\rho< \cos^{2}\!\Big(\frac{|\theta|}{2}\!-\!\frac{\pi}{4}\Big)$ for $\theta\in[-\pi,-\frac{\pi}{2}]\cup[\frac{\pi}{2},\pi]$}\Big\}
\end{equation}
with cut on the negative real axis.
\end{theorem}
${\cal C}'$ is limited by a portion of a circle and two portions of cardioids, rotated by angles $\pm\frac{\pi}{2}$ as
illustrated on figure \ref{cardioidmodified:fig}.

\begin{figure}[htb]
\begin{center}
\includegraphics[width=7cm]{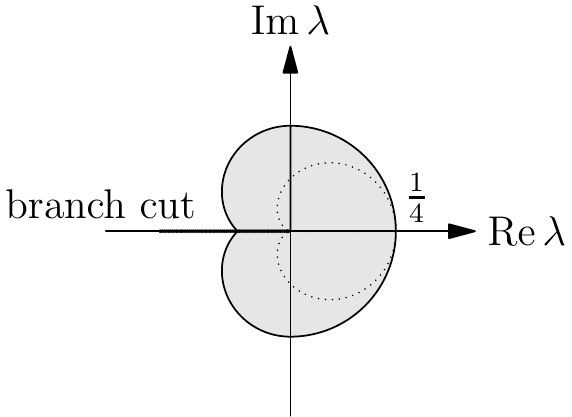}
\caption{analyticity domain for the vector model}
\end{center}
\label{cardioidmodified:fig}
\end{figure}

 Note that the analyticity domain intersects the negative real axis, on which the function has a cut because $\theta=\pm\pi$ has to be excluded.

\end{document}